\documentclass[12pt,notcite, notref]{article}
\usepackage{graphics,amsmath,amssymb,amsthm,accents}
\usepackage{epic}
\usepackage[indentafter]{titlesec}
\usepackage{multirow}
\usepackage[titletoc]{appendix}
\usepackage{booktabs}
\usepackage{array}
\usepackage{amsmath}

\makeatother \providecommand{\U}[1]{\protect \rule{.1in}{.1in}}
\usepackage[hidelinks]{hyperref}
\usepackage{xcolor}

\numberwithin{equation}{section}

\def\t#1{\widetilde{#1}}
\def\h#1{\widehat{#1}}

\def\b#1{\bar{#1}}

\newcounter{NN}
\usepackage{xcolor}

\usepackage[top=1 in,bottom=1 in,left=1.0 in,right=1.0 in]{geometry}
\usepackage{amsmath}
\usepackage{amssymb}
\usepackage{amsthm,color,cite}
\usepackage{graphicx}
\usepackage{subfigure}

\newtheorem{remark}[]{Remark}

\newtheorem{theorem}[]{Theorem}

\newtheorem{lemma}[NN]{Lemma}

\title{Connection between the symmetric discrete AKP system and bilinear ABS lattice equations}

\author{
Jing Wang$^{1,2,3}$\footnote{Email: jwang0612@shu.edu.cn; w.iac22157@kurenai.waseda.jp},~~
Da-jun Zhang$^{1,2}$\footnote{Email: djzhang@staff.shu.edu.cn},
~~Ken-ichi Maruno$^{3}$\footnote{Email: kmaruno@waseda.jp}
\\
\leftline{\small $^{1}$Department  of Mathematics, Shanghai University, Shanghai 200444, P.R. China}\\
\leftline{\small $^{2}$Newtouch Center for Mathematics of Shanghai University,  Shanghai 200444, P.R. China
}\\
\leftline{\small $^{3}$Department of Applied Mathematics, Faculty of Science and Engineering, Waseda University,}\\
\leftline{\small \quad Shinjuku-ku, Tokyo 169-8555, Japan}
}

\date{\today}

\begin{document}

\maketitle

\begin{abstract}
In this paper, we show that all the bilinear Adler-Bobenko-Suris (ABS) equations (except Q2 and Q4)
can be obtained from symmetric discrete AKP system
by taking proper reductions and continuum limits.
Among the bilinear ABS equations,
a simpler bilinear form of the ABS H2 equation is given.
In addition,  an  8-point 3-dimensional lattice equation and
an  8-point 4-dimensional lattice equation are obtained as by-products.
Both of them can be considered as extensions of the symmetric discrete AKP equation.

\begin{description}
\item[Keywords:] Adler-Bobenko-Suris equations, symmetric discrete AKP system,
bilinear equation, reduction, continuum limit
\end{description}
\end{abstract}


\section{Introduction}\label{sec-1}

There are intrinsic relations between integrable systems in different dimensions.
In discrete case, the discrete AKP (A-type Kadomtsev-Petviashvili) equation \cite{Hirota-1981,Miwa-1982},
also known as the Hirota equation or the Hirota-Miwa equation,
is fundamental in the sense that many lower dimensional (bilinear) discrete equations,
such as the discrete Korteweg-de Vries (dKdV), discrete modified KdV (dmKdV), discrete sine-Gordon,
discrete Boussinesq and Toda lattice equations, can be reduced from it \cite{DJM-1983,RGS-1992,Zabrodin-1997}.
Note that the discrete AKP (dAKP) equation is multi-dimensionally consistent \cite{ABS-2012}.
This means the high order dAKP equations of $d$-dimensions (see \cite{OHTI-1993})
can be essentially considered as  combinations of the dAKP equations defined on
any three directions in the $d$-dimensional space.
Another notable progress was made by \cite{LNSZ-2021}
in which it was shown that the dAKP $\tau$-function with reflection symmetry
can be shared by the discrete BKP (dBKP, or Miwa's equation) equation as well,
in alignment with which the result is the so-called ``KdV solves BKP'' in continuous case \cite{Alexandrov-2021}.
Note also that when the dAKP  $\tau$-function allows reflection symmetry, it is composed by the
plane wave factors (PWFs) of the
dKdV-type.
All these facts indicate that when implementing reductions of the dAKP equation,
not only the  dAKP equation itself, but also its reflected forms and high order equations
(collectively referred to as the symmetric dAKP system) and also the symmetric dBKP system
can be utilized for reductions in getting  the bilinear dKdV-type equations.

The Adler-Bobenko-Suris (ABS) equations (see \eqref{ABS-list}) provide a complete classification of
multi-dimensionally consistent affine linear quadrilateral models
with certain additional conditions \cite{ABS-2003},
which are known as H1, H2, H3, A1, A2, Q1, Q2, Q3 and Q4.
These equations have been solved in various ways
\cite{AHN-2007,AHN-2008,NAH-2009,HZ-2009,NA-IMRN-2010,AN-2010,BJ-2010,Butler-2012,ZZ-2013,SNZ-2014,
ZZ-2017,ZZ-2019}.
In fact, the PWFs of solutions (except Q4 which is the elliptic case)
and continuum limits \cite{V-SIGMA-2019,CMZ-2021}
indicate that the ABS equations are of the dKdV-type.
In this paper, we aim to derive the ABS equations (except  Q2 and Q4) in their bilinear forms
from the symmetric dAKP system.
Since A1 is connected to Q1 by $u\rightarrow (-1)^{n_1+n_2}u$ and A2 to
$\rm{Q3}_{\delta=0}$ by $u \rightarrow {u^{(-1)}}^{n_1+n_2}$,
these two equations will not be considered in this paper.
Some bilinear forms for  H1, H2, H3, Q1 and Q3 in the ABS list
involve more than one dependent variables (see \cite{HZ-2009,ZZ-2019},
while the symmetric dAKP equations exclusively involve a single $\tau$-function.
In this paper, we aim to bridge the gap between  the bilinear ABS equations
and the symmetric dAKP system.

The paper is organized as follows.
Firstly, as preliminary we list  necessary notations and the ABS equations in section\,\ref{sec-2}.
Then, in section\,\ref{sec-3}  we endeavor to reconfigure the bilinear ABS equations
using a single $\tau$-function by introducing suitable auxiliary continuous variables,
which means some functions will be explained as derivatives of the
$\tau$-function with respect to the auxiliary variables.
Correspondingly, the $\tau$-function will be presented in both Casoratian and Wronskian forms.
In section\,\ref{sec-4}, we recall the the symmetric dAKP system and their $\tau$-function.
Then in section\,\ref{sec-5}
we show what continuum limits and reductions to apply on the symmetric dAKP system
so that all the bilinear ABS equations can be obtained from reductions.
Besides, as by-products, an  8-point 3-dimensional lattice equation and
an 8-point 4-dimensional lattice equation associated with the symmetric dAKP system are found.
Appendix\,\ref{App-A} contains more details about these two equations.
Finally, conclusions are given in section\,\ref{sec-6}.

\section{Mathematical preliminary}\label{sec-2}

The functions of our interest in the fully discrete case are defined on $\mathbb{Z}^r$,
for example,
\begin{equation}
  \tau:=\tau((n_1,a_1),(n_2,a_2),\cdots,(n_r,a_r)),
\end{equation}
where $(n_1,n_2,\cdots,n_r)\in \mathbb{Z}^r$ stand  for the discrete coordinates
and $(a_1,a_2,\cdots,a_r)$ serve as spacing parameters for the corresponding directions.
By  $E_{n_i}$ we denote the forward shift operator in $n_i$ direction, e.g.,
\[E_{n_i}^{\pm j} \tau=
 \tau ((n_1,a_1),\cdots,(n_{i-1},a_{i-1}),(n_i\pm j,a_i),(n_{i+1},a_{i+1}),\cdots,(n_r,a_r)).
\]
We also introduce numerical superscripts and subscripts to denote shifts in different directions, say,
\begin{subequations}\label{tau-i}
\begin{align}
& \tau_i=E_{n_i} \tau,~~ \tau^i=E_{n_i}^{-1} \tau, ~~
\tau_{ij}=E_{n_j}E_{n_i} \tau,~~ \tau^{ij}=E_{n_j}^{-1}E_{n_i}^{-1} \tau, ~~
\tau_i^j=E_{n_j}^{-1}E_{n_i} \tau,\\
& \tau_{\,\bar i}=E_{n_1}E_{n_2}\cdots E_{n_{i-1}}E_{n_{i+1}}\cdots E_{n_r} \tau,~~ ~\cdots\cdots .
\end{align}
\end{subequations}

\begin{remark}\label{Rem-0}
To avoid any confusion, throughout the paper we only
use such notations for the functions denoted by $\tau$, $u$, $f, g$, $h$, $s$ and $v$.
\end{remark}

With these settings,
the (bilinear) dAKP equation \cite{Miwa-1982} takes the form
\begin{equation}\label{3D-dAKP-sec2}
A\equiv (a_1-a_2)\tau_3 \tau_{12}+(a_2-a_3) \tau_1\tau_{23}+(a_3-a_1) \tau_2 \tau_{13}=0,
\end{equation}
and the ABS equations are given as \cite{ABS-2003}
\begin{subequations}\label{ABS-list}
  \begin{align}
  \rm{ H }1:~~&(u-u_{12})(u_{1}-u_{2})-p+q=0, \label{H1}\\
   \mathrm{H} 2: ~~&(u-u_{12})(u_{1}-u_{2})-(p-q)(u+u_{1}+u_{2}+u_{12}+p+q)=0, \label{H2}\\
  \mathrm{H}3:~~ & p(u u_{1}+u_{2} u_{12})-q(u u_{2}+u_{1} u_{12})+\delta(p^{2}-q^{2})=0, \label{H3}\\
  \rm{A1}:~~& p(u+u_{2})(u_{1}+u_{12})-q(u+u_{1})(u_{2}+u_{12})-\delta^{2} p q(p-q)=0, \label{A1}\\
  \rm{ A2} :~~ &p(1\!-\!q^{2})(u u_{1}+u_{2} u_{12})\!-\!q(1-p^{2})(u u_{2}+u_{1} u_{12})
  -(p^{2}-q^{2})(1+u u_{1} u_{2} u_{12})=0, \label{A2}\\
  \mathrm{Q} 1:~~& p(u-u_{2})(u_{1}-u_{12})-q(u-u_{1})(u_{2}-u_{12}) +\delta^{2} p q(p-q)=0, \label{Q1}\\
  \mathrm{Q} 2:~~ & p(u-u_{2})(u_{1}-u_{12})-q(u-u_{1})(u_{2}-u_{12})\nonumber\\
  &+p q(p-q)(u+u_{1}+u_{2}+u_{12}-p^{2}+p q-q^{2})=0, \\
  \mathrm{Q} 3:~~ & p(1-q^{2})(u u_{2}+u_{1} u_{12})-q(1-p^{2})(u u_{1}+u_{2} u_{12})\nonumber\\
  &-\left(p^{2}-q^{2}\right)\left(u_{1} u_{2}+u u_{12}
  +\frac{\delta^{2}\left(1-p^{2}\right)\left(1-q^{2}\right)}{4 p q}\right)=0, \label{Q3}\\
  \mathrm{Q} 4:~~ &  \mathrm{sn} (p)(u u_{1}+u_{2} u_{12})
  -\mathrm{sn}(q)(u u_{2}+u_{1} u_{12}) -\mathrm{sn}(p-q)(u_{1} u_{2}+u u_{12}) \nonumber\\
& +\mathrm{sn}(p) \mathrm{sn}(q)
\mathrm{sn}(p-q)(1+k^{2} u u_{1} u_{2} u_{12})=0, \label{Q4}
\end{align}
\end{subequations}
where $u=u((n_1,p),(n_2,q))$ and the present form of Q4 is due to \cite{Hie}.

Next, we introduce the Wronskian and Casoratian.
Define a $N$-th order column vector
\begin{align}\label{column vector}
  \psi =(\psi_1,\psi_2,\cdots,\psi_N)^T,
\end{align}
where the entry
\begin{equation}
\psi_j:=\psi_j(n_1,n_2,\cdots,n_r,x_1,x_2,\cdots,x_l)
\end{equation}
is a function defined on $\mathbb{Z}^r\times \mathbb{R}^l$,
$\{n_j\}$ and $\{x_i\}$ serve as discrete and continuous independent variables respectively.
A Wronskian composed by the vector \eqref{column vector} and its derivatives with respect to a certain
variable $x_i$ is written as
\begin{align}
  W(\psi_1,\psi_2,\cdots,\psi_N)&:=\left\lvert \psi,\partial_{x_i}\psi,\cdots,\partial_{x_i}^{N-1}\psi \right\rvert
  \equiv\left\lvert 0,1,\cdots,N-1\right\rvert_{[x_i]}\equiv\bigl\lvert \widehat{ N-1}\bigr\rvert_{[x_i]},
\end{align}
where by the subscript $[x_i]$ we specially denote that the Wronskian is defined in terms of  $x_i$.
Note that we follow \cite{NimF-1983a} to use the compact notation $|\widehat{ N-1}|$.
Similarly, we also have
$\bigl\lvert \widehat{ N-2},N \bigr\rvert_{[x_i]}=\left\lvert 0,1,\cdots,N-2,N \right\rvert_{[x_i]}$.
Casoratians composed by the vector \eqref{column vector} with respect to the shift in $n_i$-direction
are given as
\begin{align*}
  C_1(\psi_1,\psi_2,\cdots,\psi_N)&\!:=\left\lvert \psi,E_{n_i}\psi,\cdots\!,E_{n_i}^{N-1}\psi \right\rvert
  \equiv\left\lvert 0,1,\cdots,N-1\right\rvert_{[n_i]}\equiv\bigl\lvert \widehat{ N-1}\bigr\rvert_{[n_i]},\\
  C_2(\psi_1,\psi_2,\cdots,\psi_N)&:
  =\left\lvert \psi,E_{n_i}\psi,\cdots,E_{n_i}^{N-2}\psi, E_{n_i}^{N}\psi \right\rvert\nonumber\\
  &\equiv\left\lvert 0,1,\cdots,N-2,N\right\rvert_{[n_i]}\equiv\bigl\lvert \widehat{ N-2},N\bigr\rvert_{[n_i]},\\
  C_3(\psi_1,\psi_2,\cdots,\psi_N)&:
  =\left\lvert E_{ n_i}^{-1}\psi,E_{n_i}\psi,\cdots,E_{n_i}^{N-2}\psi,E_{n_i}^{N}\psi \right\rvert \nonumber\\
  &\equiv\left\lvert -1,1,\cdots,N-2,N\right\rvert_{[n_i]}\equiv\bigl\lvert -1,\widetilde{ N-2}\bigr\rvert_{[n_i]}.
\end{align*}
Note that for a generic notation $\bigl\lvert \widehat{ N-1}\bigr\rvert_{[\lambda]}$,
it can be either a Wronskian or a Casoratian, depending on $\lambda=x_i$ or $\lambda=n_i$.

Finally, we look at bilinear derivatives and bilinear equations.
Hirota's bilinear $D$-operator is defined as
\begin{equation}\label{*}
D_x^nD_y^m f(x,y)\cdot g(x,y)
=(\partial_x-\partial_{x'})^n (\partial_y-\partial_{y'})^m f(x,y)  g(x',y')|_{x'=x,y'=y},
\end{equation}
where $f$ and $g$ are $\mathrm{C}^{\infty}$ functions with respect to $(x,y)$.
This definition can be easily extended to multidimensional cases,
while \eqref{*} serves as an illustration.
One can introduce Hirota's bilinear equation of the form
\begin{equation}
P(D_x, D_y) f(x,y)\cdot g(x,y)=0
\end{equation}
where $P(x,y)$ is a binary polynomial.
Note that such an equation has the so-called gauge property, say,
\begin{equation}
P(D_x, D_y) e^{ax+by}f(x,y)\cdot e^{ax+by} g(x,y)=e^{2(ax+by)}P(D_x, D_y) f(x,y)\cdot g(x,y)=0,
\end{equation}
where $a,b$ are constants.
A discrete bilinear equation takes a form \cite{HZ-2009}
\begin{equation}\label{*d}
  \sum_{j=1}^{\sigma}
  c_jf^{(j)}(n_1+\nu^{(j)}_1,\cdots,n_r+\nu^{(j)}_r)g^{(j)}(n_1+\mu^{(j)}_1,\cdots,n_r+\mu^{(j)}_r)=0,
\end{equation}
where $\{\nu_i^{(j)}, \mu_i^{(j)}\}$ are integers, and for each $i$, subject to
\[t_i=\nu_i^{(j)}+\mu_i^{(j)}=\nu_i^{(k)}+ \mu_i^{(k)},~~ \forall j,k\in\{1,2,\cdots, \sigma\}, \]
which guarantees the gauge property in  discrete cases \cite{HZ-2009}.
In this paper, we call $(t_1,t_2,\cdots, t_\sigma)$ to be the \emph{total index}
of the bilinear equation \eqref{*d}.

\section{Bilinear ABS equations in terms of single $\tau$-function}\label{sec-3}

Some bilinear forms for  H1, H2, H3, Q1 and Q3 in the ABS list
involve more than one dependent variables (see \cite{HZ-2009,ZZ-2019},
while the symmetric dAKP equations exclusively involve a single $\tau$-function.
In this section, we will revisit known bilinear ABS equations
and try to reconfigure them using a unified $\tau$-function and its derivatives
with respect to some auxiliary independent variables.

\subsection{H1}\label{sec-3-1}

For the H1 equation \eqref{H1},
after parameterizing $(p,q) $ by
\begin{equation}\label{para-H1}
  p= -a_1^2,~~q= -a_2^2~,
\end{equation}
we consider
\begin{equation}\label{H1-2}
  (u-u_{12})(u_{1}-u_{2})-(a_2^2-a_1^2)=0.
\end{equation}
Introducing the transformation \cite{HZ-2009}
\begin{equation}\label{trans-H1}
  u =Z-\frac{g}{f},~~ Z=a_1n_1+a_2n_2+\gamma_0,
\end{equation}
where $\gamma_0$ is an arbitrary constant, it then follows from \eqref{H1-2} that \cite{HZ-2009}
\begin{subequations}\label{be-H1}
  \begin{align}
    &g_{2}   f_{1}-g_{1}   f_{2}+(a_1-a_2)(   f_{1}   f_{2}-f   f_{12})=0,\\
    &g   f_{12}-g_{12}f+(a_1+a_2)(f   f_{12}-   f_{1}   f_{2})=0.
  \end{align}
\end{subequations}
The Casoratian solutions to the above bilinear H1 equation \eqref{be-H1} are \cite{HZ-2009}
\begin{equation}\label{sol-Casoratian-be-H1}
  f=\bigl\lvert\widehat{N-1} \bigr\rvert_{[l_1]},~~g=\bigl\lvert\widehat{N-2},N \bigr\rvert_{[l_1]},
\end{equation}
composed by $\psi=(\psi_1, \psi_2, \cdots, \psi_N)^T$ with
\begin{align}\label{3.7}
  \psi_j(n_1,n_2,l_1;k_j) =\rho_j^+k_j^{l_1}(a_1+k_j)^{n_1}(a_2+k_j)^{n_2}
  +\rho_j^-(-k_j)^{l_1}(a_1-k_j)^{n_1}(a_2-k_j)^{n_2}.
\end{align}

To write the bilinear H1 equation \eqref{be-H1} using a single $\tau$-function,
due to the arbitrariness of $\rho_j^{\pm}$,
we can replace $\rho_j^{\pm}$ in \eqref{3.7} with $\rho_j^{\pm} e^{\pm k_j x}$
and consider the following entry in the vector $\psi$:
\begin{align}\label{psi-H1}
  \psi_j(n_1,n_2,x;k_j)= \rho_j^+e^{k_jx}k_j^{l_1}(a_1+k_j)^{n_1}(a_2+k_j)^{n_2}
  +\rho_j^-e^{-k_jx}(-k_j)^{l_1}(a_1-k_j)^{n_1}(a_2-k_j)^{n_2}.
\end{align}
Then we have the following.

\begin{lemma}\label{Lem-3-1}
For the Casoratians $f$ and $g$ defined in \eqref{sol-Casoratian-be-H1}
composed by $\psi$ with entries \eqref{psi-H1},
the following relations hold:
\begin{subequations}
\begin{align}
&f=\bigl\lvert\widehat{N-1} \bigr\rvert_{[l_1]}=\bigl\lvert\widehat{N-1} \bigr\rvert_{[x]}
=\bigl\lvert\widehat{N-1} \bigr\rvert_{[n_\sigma]},~~ (\sigma=1,2), \\
&g=\bigl\lvert\widehat{N-2},N \bigr\rvert_{[l_1]}=\bigl\lvert\widehat{N-2},N \bigr\rvert_{[x]}=f_x
=\bigl\lvert\widehat{N-2},N \bigr\rvert_{[n_\sigma]}-a_{\sigma} N f,~~  (\sigma=1,2). \label{g-relation-H1}
\end{align}
\end{subequations}
\end{lemma}

\begin{proof}
From \eqref{psi-H1}  we can find that
\begin{subequations}\label{shift-relation}
      \begin{align}
        &E_{l_1}\psi_j=\partial_{x}\psi_j,~~E_{l_1}^{i}\psi_j=\partial_{x}^i\psi_j,\label{l1-x-H1}\\
        &E_{l_1}\psi_j=(E_{n_{\sigma}}-a_{\sigma})\psi_j,~~
        E_{l_1}^i\psi_j=(E_{n_{\sigma}}-a_{\sigma})^i\psi_j,~~\sigma=1,2.\label{l1-v-H1}
\end{align}
\end{subequations}
The relation \eqref{l1-x-H1} immediately gives rise to  $f=\bigl\lvert\widehat{N-1} \bigr\rvert_{[x]}$.
Using \eqref{l1-v-H1} we have
\begin{align*}
 f&=\bigl\lvert\widehat{N-1} \bigr |_{[l_1]}
 =\bigl\lvert \psi,E_{l_1}\psi,\cdots, E_{l_1}^{N-1}\psi \bigr\rvert \\
      &=\left\lvert \psi,(E_{n_{\sigma}}-a_{\sigma})\psi,\cdots,(E_{n_{\sigma}}-a_{\sigma})^{N-1}\psi \right\rvert \\
      &=\bigl\lvert \psi,E_{n_{\sigma}}\psi,\cdots, E_{n_{\sigma}}^{N-1}\psi \bigr\rvert \\
      &=\bigl\lvert \widehat{N-1} \bigr\rvert_{[n_{\sigma}]},~~~ (\sigma=1,2).
\end{align*}
Similarly, one can obtain \eqref{g-relation-H1}.

\end{proof}

This Lemma immediately indicates that the Casoratians
\begin{equation}\label{sol-Casoratian-be-H1-1}
f=\bigl\lvert\widehat{N-1} \bigr\rvert_{[\lambda]},~~g=\bigl\lvert\widehat{N-2},N \bigr\rvert_{[\lambda]},
~~(\lambda=n_1, n_2, x),
\end{equation}
composed by $\psi$ with entries \eqref{psi-H1}
are solutions to the bilinear H1 equation \eqref{be-H1}.
In addition, noting that $g=f_x$, we are led to the following.

\begin{theorem} \label{th-1}
Through the transformation
\begin{equation}\label{trans-H1-f}
  u =Z-\frac{f_x}{f}, ~~~ Z=a_1n_1+a_2n_2+\gamma_0,
\end{equation}
the H1 equation \eqref{H1-2} allows
the following bilinear form
with a single $\tau$-function $f$:
\begin{subequations}\label{be-f-H1}
    \begin{align}
      &D_x    f_{2}\cdot   f_{1}+(a_1-a_2)(   f_{1}   f_{2}-f   f_{12})=0, \label{be-f-H1a}\\
      &D_x  f\cdot   f_{12}+(a_1+a_2)(f   f_{12}-   f_{1}   f_{2})=0,  \label{be-f-H1b}
    \end{align}
\end{subequations}
and its solutions are given by
\begin{equation}\label{sol-Caso-Wrons-be-H1-2}
  f=\bigl\lvert\widehat{N-1} \bigr\rvert_{[\lambda]}, ~~(\lambda=n_1, n_2, x),
\end{equation}
composed by $\psi=(\psi_1,\psi_2,\cdots,\psi_N)^T$ with the entries $\psi_j$ given in \eqref{psi-H1}.
Note that the terms $(\pm k_j)^{l_1}$ in \eqref{psi-H1} can be absorbed into
$\rho_j^{\pm}$.
\end{theorem}

\subsection{H2}\label{sec-3-2}

H2 equation \eqref{H2} with the parametrization \eqref{para-H1} is
\begin{equation}\label{H2-1}
  (u-u_{12})(u_{1}-u_{2})+(a_1^2-a_2^2)(u+u_{1}+u_{2}+u_{12}-a_1^2-a_2^2)=0.
\end{equation}
This equation allows a bilinearization with 5 equations (see equation (4.21) in \cite{HZ-2009}),
while in the following we provide a simpler form.
By the transformation
\begin{equation}
  u=Z^2-2Z\frac{g}{f}+\frac{h}{f},~~
  Z=a_1n_1+a_2n_2+\gamma_0,
\end{equation}
where $\gamma_0$ is a constant,
H2 equation \eqref{H2-1} is bilinearized as\footnote{Note that this bilinear form
has been obtained in  \cite{CMZ-2021} using the connection between H1 and H2.
However, in this paper we find a direct formulation of the bilinearization, see \eqref{3.17}.}
\begin{subequations}\label{be-f-g-h-H2}
  \begin{align}
  &\mathcal{H}_1\equiv2a_1(   f_{1}g-fg_{1})-2gg_{1}+ h_1f+h   f_{1}=0,\\
  &\mathcal{H}_2\equiv2a_2(   f_{2}g-fg_{2})-2gg_{2}+ h_2f+h   f_{2}=0,
  \end{align}
\end{subequations}
which are connected with \eqref{H2-1} as
\begin{equation}\label{3.17}
  \mathrm{H}_2=P_1 \mathcal{H}_1+P_2 \mathcal{H}_2+P_3 \mathcal{\h H}_1
  +P_4 \mathcal{\t H}_2+P_5 \mathcal{H}_1\mathcal{\h H}_1+P_6 \mathcal{H}_2\mathcal{\t H}_2
\end{equation}
with
\begin{align*}
  &P_1=\frac{2}{f f_1}\left(Z+a_2-\frac{g_{2}}
  {   f_{2}}\right)\left(Z+a_1+a_2-\frac{g_{12}}{   f_{12}}\right),\\
&P_2=\frac{-2}{f f_2}\left(Z+a_1-\frac{g_{1}}
{   f_{1}}\right)\left(Z+a_1+a_2-\frac{g_{12}}{   f_{12}}\right),\\
&P_3=\frac{2}{f f_1}\left(Z-\frac{g}{f}\right)\left(Z+a_1-\frac{g_{1}}{f_{1}}\right),\\
&P_4=\frac{-2}{f f_2}\left(Z-\frac{g}{ f}\right)\left(Z+a_2-\frac{g_{2}}
{   f_{2}}\right),\\
&P_5=\frac{1}{f   f_{1}   f_{2}   f_{12}}, ~~~
P_6=\frac{-1}{f   f_{1}   f_{2}   f_{12}},
\end{align*}
where the $\mathcal{\h H}_1$ (or $\mathcal{\t H}_2$) denotes that a $n_2$-direction (or $n_1$-direction)
forward shift on the whole expression $\mathcal{H}_1$ (or $\mathcal{H}_2$).

As for solutions to \eqref{be-f-g-h-H2}, we have the following (see \cite{CMZ-2021}).

\begin{lemma}\label{Lem-3-2}
The bilinear H2 equation \eqref{be-f-g-h-H2} admits Casoratian solutions
\begin{align}\label{f-g-h-be-H2}
f=\bigl\lvert\widehat{N-1}\bigr\rvert_{[l_1]}, ~~g=\bigl\lvert\widehat{N-2},N\bigr\rvert_{[l_1]}, ~~
h=\bigl\lvert\widehat{N-3},N-1,N\bigr\rvert_{[l_1]}+\bigl\lvert\widehat{N-2},N+1\bigr\rvert_{[l_1]},
\end{align}
composed by $\psi=(\psi_1,\psi_2\cdots,\psi_N)^T$ with entries $\psi_j$ given in \eqref{3.7}.
\end{lemma}

Similar to the treatment to H1, we introduce $\psi_j$ defined by \eqref{psi-H1}.
Then, it follows that the relations $g=f_x$ and $h=f_{xx}$ hold and we obtain the following Theorem.

\begin{theorem}\label{th-2}
By the transformation
\begin{equation}
  u=Z^2-2Z\frac{f_x}{f}+\frac{f_{xx}}{f},
  ~~ Z=a_1n_1+a_2n_2+\gamma_0,
\end{equation}
H2 equation \eqref{H2} allows the following bilinear form
\begin{subequations}\label{be-f-H2}
    \begin{align}
      &(D_x^2+2a_1D_x) f\cdot    f_{1}=0,\\
      &(D_x^2+2a_2D_x) f\cdot    f_{2}=0,
    \end{align}
  \end{subequations}
which has a solution
\begin{equation}\label{sol-Caso-Wrons-be-H2}
    f=\bigl\lvert\widehat{N-1} \bigr\rvert_{[\lambda]}, ~~(\lambda=n_1, n_2, x)
\end{equation}
composed by $\psi=(\psi_1,\psi_2,\cdots,\psi_N)^T$ with the entries $\psi_j$  given in \eqref{psi-H1}.
\end{theorem}

\subsection{H3}\label{sec-3-3}

For H3 equation \eqref{H3}, i.e.,
\begin{equation}\label{H3-1}
  p(uu_{1}+u_{2}u_{12})-q(uu_{2}+u_{1}u_{12})-\delta(q^2-p^2)=0,
\end{equation}
it allows two bilinear forms, both of which consist of a single $\tau$-function \cite{HZ-2009}.
With the parametrisation \cite{HZ-2009}
\begin{equation}
  p^2=\frac{r^2a_3^2}{a_3^2-a_1^2},~~q^2=\frac{r^2a_3^2}{a_3^2-a_2^2},~~
  \alpha^2=-\frac{a_1-a_3}{a_1+a_3},~~\beta^2=-\frac{a_2-a_3}{a_2+a_3},~~c_1c_2=-\frac{1}{4}r\delta
\end{equation}
and the transformation \cite{HZ-2009}
\begin{equation}
  u =c_1\alpha^{n_1}\beta^{n_2}\frac{   f_{3}}{f}+c_2\alpha^{-n_1}\beta^{-n_2}\frac{f^{3}}{f}~,
\end{equation}
H3 can be bilinearized as  either \cite{HZ-2009}
\begin{subequations}\label{be-H3-f-1}
  \begin{align}
    &(a_2+a_3)   f_{12}   f_{3}+(a_1-a_3)ff_{123}-(a_1+a_2)   f_{1}   f_{23}=0,\label{be-H3-1}\\
    &(a_3-a_2)   f_{12}f^{3}-(a_1+a_3)ff_{12}^{3}+(a_1+a_2)   f_{1}f_2^{3}=0,\label{be-H3-2}\\
    &2a_3(a_1-a_2)f   f_{12}+(a_3-a_1)(a_2+a_3)   f_{13}f_2^{3}+(a_1+a_3)(a_2-a_3)   f_{23}f_{1}^{3}=0,
    \label{be-H3-3}
  \end{align}
\end{subequations}
or
\begin{subequations}\label{be-H3-f-2}
  \begin{align}
&2a_3f    f_{1}+(a_1-a_3)   f_{13}f^{3}-(a_1+a_3)   f_{3}f_1^{3}=0,\label{be-H3-1-1}\\
&2a_3f    f_{2}+(a_2-a_3)   f_{23}f^{3}-(a_2+a_3)   f_{3}f_2^{3}=0.\label{be-H3-1-2}
  \end{align}
\end{subequations}
Both forms have a solution
\begin{equation}\label{f-H3}
  f=\bigl\lvert\widehat{N-1}\bigr\rvert_{[\lambda]},~~ (\lambda=n_1, n_2, n_3)
\end{equation}
composed by $\psi=(\psi_1,\psi_2,\cdots,\psi_N)^T$ with entries
\begin{equation}
\label{psi-H3}
  \psi_j(n_1,n_2,n_3;k_j)=\rho_j^+\prod^3_{i=1}(a_i+k_j)^{n_i}+\rho_j^{-}\prod^3_{i=1}(a_i-k_j)^{n_i}.
\end{equation}

\subsection{Q1}\label{sec-3-4}

\subsubsection{Q1 with power background}\label{sec-3-4-1}

For Q1 equation \eqref{Q1}, i.e.,
\begin{equation}\label{Q1-1}
   p(u-u_{2})(u_{1}-u_{12})-q(u-u_{12})(u_{2}-u_{12})-\delta^2 p q (q-p)=0,
  \end{equation}
after the parametrization \cite{HZ-2009}
  \begin{equation}
    p=\frac{ra_3^2}{a_1^2-a_3^2},~~q=\frac{ra_3^2}{a_2^2-a_3^2},~~\alpha=\frac{a_1-a_3}{a_1+a_3},
    ~~\beta=\frac{a_2-a_3}{a_2+a_3},~~c_1c_2=\frac{1}{16}r^2\delta^2
  \end{equation}
and through the transformation
\begin{equation}
  u=c_1\alpha^{n_1}\beta^{n_2}\frac{  f_{33}}{f}+c_2\alpha^{-n_1}\beta^{-n_2}\frac{f^{33}}{f}~,
\end{equation}
it can be bilinearized as \cite{HZ-2009}
\begin{subequations}\label{be-f-Q1-1}
  \begin{align}
&2a_3f    f_{1}+(a_1-a_3)   f_{13}f^{3}-(a_1+a_3)   f_{3}f_1^{3}=0,\label{be-f1-Q1-1}\\
&2a_3f    f_{2}+(a_2-a_3)   f_{23}f^{3}-(a_2+a_3)   f_{3}f_2^{3}=0,\label{be-f1-Q1-2}
  \end{align}
\end{subequations}
which is nothing but \eqref{be-H3-f-2} and has a solution $f$ given in \eqref{f-H3}
composed by $\psi$ with entries \eqref{psi-H3}.

\subsubsection{Q1 with linear background}\label{sec-3-4-2}

Q1 equation \eqref{Q1-1} can have a linear background solution
$u =\alpha n_1+\beta n_2+\gamma_0$
under the parametrization \cite{HZ-2009}
\begin{equation}\label{3.33}
  p=\frac{c^2/r-\delta^2r}{a_1^2-\delta^2},~~q=\frac{c^2/r-\delta^2r}{a_2^2-\delta^2},
  ~~\alpha=pa_1,~~\beta=qa_2,
\end{equation}
where $c,r$ and $\gamma_0$ are all constants.
Through the transformation
\begin{equation}\label{trans-Q1-2}
  u =\alpha n_1+\beta n_2+\gamma_0-(c^2/r-\delta^2r)\frac{g}{f},
\end{equation}
there exists the following bilinearization \cite{HZ-2009}
\begin{subequations} \label{be-f-g-Q1}
\begin{align}
&\mathcal{Q}_1^{(1)}\equiv (a_2-\delta)f_{123}f+(a_1+\delta)   f_{12}   f_{3}-(a_1+a_2)   f_{13}   f_{2}=0,\\
&\mathcal{Q}_2^{(1)}\equiv(a_1-a_2)f_{123} f+(a_2+\delta)   f_{13}   f_{2}-(a_1+\delta)   f_{23}   f_{1}=0,\\
&\mathcal{Q}_3^{(1)}\equiv   f_{1}   f_{23}-   f_{13}   f_{2}+(\delta-a_1)
 f_{13}g_{2}+(a_2-\delta)   f_{23}g_{1}+(a_1-a_2)   f_{3}g_{12}=0,\\
&\mathcal{Q}_4^{(1)}\equiv   f_{23}   f_{1}-   f_{13}   f_{2}-(a_1+\delta)
f_{23}g_{1}+(a_2+\delta)   f_{13}g_{2}+(a_1-a_2)f_{123}g=0,
\end{align}
\end{subequations}
with solutions
\begin{equation}\label{sol-f-g-Q1-1}
  f=\bigl\lvert \widehat{N-1} \bigr\rvert_{[n_3]}, ~~g=\bigl\lvert -1,\widetilde{N-1} \bigr\rvert_{[n_3]},
\end{equation}
composed by $\psi$ with entries
\begin{equation}
  \psi_j(n_1,n_2,n_3;k_j)=\rho_j^+(a_1+k_j)^{n_1}(a_2+k_j)^{n_2}(\delta+k_j)^{n_3}
  +\rho_j^-(a_1-k_j)^{n_1}(a_2-k_j)^{n_2}(\delta-k_j)^{n_3}.
\end{equation}

In order to express the bilinear form using a single $\tau$-function,
we first introduce
\begin{equation}
\label{psi-Q1-2}
  \psi_j(n_1,n_2,n_3,y;k_j)=\,\rho_j^+e^{\frac{y}{a_3+k_j}}\prod^3_{i=1}(a_i+k_j)^{n_i}
   +\rho_j^-e^{\frac{y}{a_3-k_j}}\prod^3_{i=1}(a_i-k_j)^{n_i}
\end{equation}
with $a_3=\delta$.
Thus we have $g=f_y$
and then we can rewrite $\mathcal{Q}_3^{(1)}$ and $\mathcal{Q}_4^{(1)}$ as
\begin{subequations} \label{f-Q1-1}
  \begin{align}
  &\mathcal{Q}_3^{(2)}\equiv-   f_{13}   f_{2}+   f_{1}   f_{23}+(a_3-a_1)   f_{13}   f_{2,y}
  +(a_2-a_3)   f_{23}   f_{1,y}+(a_1-a_2)   f_{3}   f_{12,y}=0,\label{be-f-Q1-1-3}\\
  &\mathcal{Q}_4^{(2)}\equiv   f_{23}   f_{1}-   f_{13}   f_{2}-(a_1+a_3)   f_{23}   f_{1,y}
  +(a_2+a_3)   f_{13}   f_{2,y}+(a_1-a_2)f_{123}f_y=0\label{be-f-Q1-1-4}.
  \end{align}
  \end{subequations}
Note that these two equations are not ready to be written
in terms of  Hirota's $D$-operator.
However, using the known bilinear equations (see \eqref{3D-dAKP-sec2})
\begin{equation}\label{3D-dAKP-sec2-f}
\mathcal{Q}_0\equiv A \equiv  (a_1-a_2)f_3 f_{12}+(a_2-a_3) f_1f_{23}+(a_3-a_1) f_2 f_{13}=0,
\end{equation}
we have the following relations
  \begin{equation}
  \mathcal{Q}_3=  2 \mathcal{Q}_3^{(2)}-\mathcal{Q}_{0,y}=0,~~
  \mathcal{Q}_4=2 \mathcal{Q}_4^{(2)}-\mathcal{Q}^{(1)}_{2,y}=0,
  \end{equation}
which give rise to the expected form
\begin{subequations}
  \begin{align}
   &\!2   f_{1}   f_{23}-\!2   f_{13}   f_{2}\!+\!(a_3-a_1)D_y   f_{2}\cdot   f_{13}\!
   +\!(a_2-a_3)D_y   f_{1}\cdot   f_{23}\!-\!(a_1-a_2)D_y   f_{3}\cdot   f_{12}\!=0,\\
  &2   f_{23}   f_{1}-2   f_{13}   f_{2}\!-\!(a_1+a_3)D_y   f_{1}\cdot   f_{23}\!
  +\!(a_2+a_3)D_y   f_{2}\cdot   f_{13}\!+\!(a_1-a_2)D_yf\cdot f_{123}=0.
  \end{align}
  \end{subequations}
Thus we arrive at the following.

\begin{theorem}\label{th-3}
For Q1 equation \eqref{Q1-1} with parametrization \eqref{3.33}, through the transformation
\begin{equation}\label{trans-Q1-2y}
  u =\alpha n_1+\beta n_2+\gamma_0-(c^2/r-\delta^2r)\frac{f_y}{f},
\end{equation}
it has bilinear equations (with $a_3=\delta$)
\begin{subequations} \label{be-f-Q1-2}
    \begin{align}
    & \mathcal{Q}_0\equiv  (a_1-a_2)f_3 f_{12}+(a_2-a_3) f_1f_{23}
    +(a_3-a_1) f_2 f_{13}=0, \\
    &\mathcal{Q}_1\equiv(a_2-a_3)f_{123}f+(a_1+a_3)   f_{12}   f_{3}-(a_1+a_2)   f_{13}
      f_{2}=0,\label{be-Q1-1}\\
    &\mathcal{Q}_2\equiv(a_1-a_2)f_{123} f+(a_2+a_3)   f_{13}   f_{2}-(a_1+a_3)   f_{23}
      f_{1}=0,\label{be-Q1-2}\\
    &\mathcal{Q}_3\equiv2   f_{1}   f_{23}\!-\!2   f_{13}   f_{2}\!+\!(a_3\!-\!a_1)D_y   f_{2}\!\cdot \!
    f_{13}\!+\!(a_2\!-\!a_3)D_y   f_{1}\!\cdot \!  f_{23}\!-\!(a_1\!-\!a_2)D_y   f_{3}\!\cdot \!
    f_{12}=0,\label{be-Q1-3}\\
    &\mathcal{Q}_4\equiv 2   f_{23}   f_{1}\!-\!2   f_{13}   f_{2}\!-\!(a_1\!+\!a_3)D_y   f_{1}\!\cdot \!  f_{23}\!
    +\!(a_2\!+\!a_3)D_y   f_{2}\!\cdot \!  f_{13}\!+\!(a_1\!-\!a_2)D_yf\!\cdot\! f_{123}=0\label{be-Q1-4},
    \end{align}
    \end{subequations}
which allow a solution
\begin{equation}\label{sol-f-Q1-2}
      f=\bigl\lvert \widehat{N-1} \bigr\rvert_{[\lambda]},~~ (\lambda=n_1, n_2 ,n_3),
\end{equation}
composed by $\psi=(\psi_1,\psi_2,\cdots,\psi^N)^T$ with  entries \eqref{psi-Q1-2}.
\end{theorem}

\subsection{Q3}\label{sec-3-5}

For Q3 equation \eqref{Q3}, via the parametrization
\begin{equation}
  p^{2}=\frac{a_1^{2}-a_4^{2}}{a_1^{2}-a_3^{2}}, \quad
  q^{2}=\frac{a_2^{2}-a_4^{2}}{a_2^{2}-a_3^{2}}, \quad
  P=\frac{\left(a_4^{2}-a_3^{2}\right)p}{1-p^{2}}, \quad Q=\frac{\left(a_4^{2}-a_3^{2}\right) q}{1-q^{2}},
\end{equation}
and replacing $u$ with $(a_3^2-a_4^4)u$, the Q3 in \eqref{Q3} is written as \cite{AHN-2008,NAH-2009}
\begin{equation}\label{Q3-a1-a2}
  P(u u_{2}+u_{1} u_{12})-Q(u u_1+u_{2} u_{12})-\left(a_1^{2}-a_2^{2}\right)(u_{2} u_1+u u_{12})
  =\delta^{2} \frac{\left(a_1^{2}-a_2^{2}\right)}{4 P Q}.
\end{equation}
Solutions to this equation can be given through those of the Nijhoff-Quispel-Capel (NQC) equation,
which reads \cite{NQC-1983}
\begin{equation}\label{NQC}
  \!\!\!\!\!\!\frac{1\!+\!(a_1\!-\!a_3)s(a_3,a_4)\!-\!(a_1\!+\!a_4)s_1(a_3,a_4)}
  {1\!+\!(a_2\!-\!a_3)s(a_3,a_4)\!-\!(a_2\!+\!a_4) s_2(a_3,a_4)}\!
  =\!\frac{1\!+\!(a_2\!-\!a_4)s_1(a_3,a_4)\!-\!(a_2\!+\!a_3)s_{12}(a_3,a_4)}
  {1\!+\!(a_1\!-\!a_4)s_2(a_3,a_4)\!-\!(a_1\!+\!a_3)s_{12}(a_3,a_4)},
\end{equation}
where $a_3$ and $a_4$ are two extra parameters, $s(a_3,a_4)$ is a function of $(n_1, n_2)$
and it holds specially that
\begin{equation}
s(a_3,a_4)=s(a_4,a_3).
\label{S-34}
\end{equation}
Note that here and below in \eqref{S-V} $s_1, s_2, s_{12}, v_1, v_2$, etc.
are defined as in \eqref{tau-i}, see Remark \ref{Rem-0}.
One can also consider that $s(a_3,a_4)$ is a function of $(n_1, n_2,n_3,n_4)$
with $(a_1,a_2,a_3,a_4)$ as spacing parameters,
while $n_3$ and $n_4$ are dummy variables that do not appear in the NQC equation \eqref{NQC}.
Solution of the Q3 equation \eqref{Q3-a1-a2} can be expressed as   follows  \cite{NAH-2009}
\begin{equation}\label{sol-Q3}
  \begin{aligned}
   u=&c_1F(a_3,a_4)[1-(a_3+a_4)s(a_3,a_4)]+c_2F(a_3,-a_4)[1-(a_3-a_4)s(a_3,-a_4)]\\
&+c_3F(-a_3,a_4)[1\!+\!(a_3\!-\!a_4)s(-a_3,a_4)]+c_4F(-a_3,-a_4)[1+(a_3+a_4)s(-a_3,-a_4)],
  \end{aligned}
\end{equation}
where  the function $F(a_3,a_4)$ is defined as
\begin{equation}\label{F}
  F(a_3,a_4)=\left(\frac{(a_1+a_3)(a_1+a_4)}{(a_1-a_3)(a_1-a_4)}\right)^{\frac{n_1}{2}}
  \left(\frac{(a_2+a_3)(a_2+a_4)}{(a_2-a_3)(a_2-a_4)}\right)^{\frac{n_2}{2}}
\end{equation}
and the constants $c_i,~i=1\cdots 4$, subject to
\begin{equation}\label{ABCD}
  c_1c_4(a_3+a_4)^2-c_2c_3(a_3-a_4)^2=-\frac{\delta^2}{16a_3a_4}.
\end{equation}
Thus, the bilinear form of Q3 can be obtained from the bilinearization of the NQC equation \eqref{NQC} \cite{ZZ-2019}.

Note that the NQC equation can be derived from the following equation set \cite{NQC-1983,NAH-2009,ZZ-2013,HJN-2016}
\begin{subequations}\label{S-V}
\begin{align}
 & 1+(a_1-a_3)s(a_3,a_4)-(a_1+a_4)s_1(a_3,a_4)=v_1(a_3)v(a_4),\\
 &1+(a_2-a_3)s(a_3,a_4)-(a_2+a_4)s_2(a_3,a_4)=v_2(a_3)v(a_4),\\
 & 1+(a_1-a_4)s(a_3,a_4)-(a_1+a_3)s_1(a_3,a_4)=v_1(a_4)v(a_3),\\
 &1+(a_2-a_4)s(a_3,a_4)-(a_2+a_3)s_2(a_3,a_4)=v_2(a_4)v(a_3),
\end{align}
\end{subequations}
together with symmetry \eqref{S-34}
and assumption that $v(a_3)$ is a function related
to $a_3$ but independent of $a_4$, i.e.,
$v(a_4)=v(a_3)|_{a_3\rightarrow a_4}$.
One can bilinearize the above system via the transformation
\begin{equation}
  s(a_3,a_4)=\frac{1}{a_3+a_4}\biggl(1-\frac{f^{34}}{f}\biggr)\biggr|_{n_3=n_4=0},
  ~~v(a_3)=\frac{f^3}{f}\biggr|_{n_3=n_4=0},~~v(a_4)=\frac{f^4}{f}\biggr|_{n_3=n_4=0},
\end{equation}
and get the bilinear form of the NQC equation as well as of the Q3 equation: (cf.\cite{ZZ-2019})
\begin{subequations}\label{be-Q3-1}
  \begin{align}
    &(a_1-a_3)f f_{134}-(a_1+a_4)f_1f_{34}+(a_3+a_4)f_{14}f_3=0,\\
    &(a_2-a_3)f f_{234}-(a_2+a_4)f_2 f_{34}+(a_3+a_4)f_{24}f_3=0,\\
    &(a_2-a_4)ff_{234}-(a_2+a_3)f_2f_{34}+(a_3+a_4)f_4f_{23}=0,\\
    &(a_1-a_4)ff_{134}-(a_1+a_3)f_1f_{34}+(a_3+a_4)f_4f_{13}=0.
  \end{align}
\end{subequations}
This bilinear system involves only the function $f$
and admits the solution in Casoratian form \cite{ZZ-2019}:
\begin{equation}\label{sol-f-Q3}
  f=\bigl\lvert\widehat{N-1}\bigr\rvert_{[\lambda]}, ~~(\lambda=n_1,n_2,n_3,n_4),
\end{equation}
composed by $\psi=(\psi_1,\psi_2,\cdots,\psi_N)^T$ with entries
  \begin{equation}\label{psi-Q3}
  \psi_j(n_1,n_2,n_3,n_4;k_j)= \rho_j^+\prod^4_{i=1}(a_i+k_j)^{n_i}
      +\rho_j^-\prod^4_{i=1}(a_i-k_j)^{n_i}.
    \end{equation}
One can check that the above defined $s$ and $v$ satisfy \eqref{S-34} and $v(a_4)=v(a_3)|_{a_3\rightarrow a_4}$.

In conclusion, we have the following for Q3.

\begin{theorem}\label{th-4}
Q3 equation \eqref{Q3} can be bilinearized as the system \eqref{be-Q3-1} via the transformation
\begin{align}
u=
\left[c_1F(a_3,a_4)\frac{f^{34}}{f}+c_2F(a_3,-a_4)\frac{f_{4}^3}{f}
+c_3F(-a_3,a_4)\frac{f_3^{4}}{f}+c_4F(-a_3,-a_4)\frac{f_{34}}{f}\right]
\biggr|_{n_3=n_4=0},
\end{align}
where $F$ is defined as \eqref{F} and the
constants $c_1,c_2,c_3,c_4$
are subject to \eqref{ABCD}.
The bilinear form allows the Casoratian solution \eqref{sol-f-Q3}
composed by $\psi$ with entries \eqref{psi-Q3}.
\end{theorem}

\subsection{Universal $\tau$-function for the bilinear ABS} \label{sec-3-6}

Let us quickly summarize  this section about the universal $\tau$-function for the bilinear ABS equations.

\begin{theorem}\label{th-5}
The bilinear ABS equations, including \eqref{be-f-H1}, \eqref{be-f-H2}, \eqref{be-H3-f-1},
\eqref{be-H3-f-2}, \eqref{be-f-Q1-1}, \eqref{be-f-Q1-2} and \eqref{be-Q3-1},
share the same  single $\tau$-function
\begin{equation}\label{sol-f-ABS}
    f=\bigl\lvert \widehat{N-1} \bigr\rvert_{[\lambda]}, ~~(\lambda=n_1,n_2,n_3,n_4,x)
\end{equation}
composed by $\psi=(\psi_1,\psi_2,\cdots,\psi_N)^T$ with universal entries
  \begin{equation}
\label{psi-ABS}
 \!\!\psi_j(n_1,n_2,n_3,n_4,x,y;k_j) =
 \rho_j^+e^{k_jx+\frac{y}{a_3+k_j}}\prod^4_{i=1}(a_i+k_j)^{n_i}
+\rho_j^-e^{-k_jx+\frac{y}{a_3-k_j}}\prod^4_{i=1}(a_i-k_j)^{n_i}.
\end{equation}
\end{theorem}

\section{Symmetric dAKP system}\label{sec-4}

In the previous section we have expressed the bilinear forms of H1, H2, H3, Q1 and Q3
using a single $\tau$-function in Wronskian-Casoratian form, $f=\big|\widehat{N-1} \bigr |_{[\lambda]}$,
composed by a universal vector $\psi$ (see Theorem \ref{th-5}).
To explain how those bilinear ABS equations are related to the dAKP equation and its deformations,
in the following we will have a close look at the dAKP equations and its reflected forms, which will be used
in reduction.

Consider the following (high order dAKP)  compact form \cite{OHTI-1993}
\begin{equation}\label{3-4-5-D-AKP}
 \left | \begin{matrix}
    1 &a_1 & a_1^2 & \cdots & a_1^{\sigma-2}  &\tau_1 \tau_{\bar 1} \\
    1 &a_2 & a_2^2 &\cdots &a_2^{\sigma-2}&\tau_2\tau_{\b2}\\
    \vdots  &\vdots &\vdots& &\vdots &\vdots\\
    1 &a_{\sigma}&a_{\sigma}^2 &\cdots &a_{\sigma}^{\sigma-2} &\tau_{\sigma}\tau_{\b \sigma}\\
    \end{matrix} \right |=0, ~~ (\mathrm{for} ~\sigma=3,4,5),
  \end{equation}
where the notation $\tau_{\b \sigma}$ can be referred from \eqref{tau-i} in  Sec.\ref{sec-2}.
When $\sigma=3$, it is the dAKP equation \eqref{3D-dAKP-sec2}, i.e.,
\begin{equation}\label{3D-dAKP}
  A\equiv (a_1-a_2)\tau_3 \tau_{12}+(a_2-a_3) \tau_1\tau_{23}+(a_3-a_1) \tau_2 \tau_{13}=0,
    \end{equation}
and when $\sigma=4$, it gives the 4D dAKP equation
  \begin{align}\label{4D-dAKP}
\mathfrak{A}\equiv & (a_2-a_3)(a_3-a_4)(a_2-a_4)\tau_1\tau_{234}
    -(a_1-a_3)(a_3-a_4)(a_1-a_4)\tau_2\tau_{134}\nonumber\\
    &+(a_1-a_2)(a_2-a_4)(a_1-a_4)\tau_3\tau_{124}-(a_1-a_2)(a_2-a_3)(a_1-a_3)\tau_4\tau_{123}=0,
  \end{align}
which can be considered as a result of the 4D consistency of the dAKP equation \eqref{3D-dAKP} \cite{LNSZ-2021}.
The  5D dAKP equation is the case for $\sigma=5$, which we do not write out here the expanded form.
Note that the dAKP equation is 4D consistent \cite{ABS-2012},
and we believe the 5D dAKP equation is also a result of the multidimensional consistency of the dAKP equation.

Solutions  of \eqref{3-4-5-D-AKP} can be given either in Hirota's exponential-polynomial form
\begin{equation}\label{tau-solution}
  \tau=\sum_{\mu=0,1} \exp \left[\sum_{j=1}^{N} \mu_{j} \eta_{j}
  +\sum_{1 \leqslant i<j}^{N} \mu_{i} \mu_{j} a_{i j}\right]
\end{equation}
with the PWF and phase factor
\begin{equation}\label{PWF-dAKP}
  \rho=\mathrm{e}^{\eta_{j}}=\prod_{r=1}^{\sigma} \left(\frac{a_r-q_{j}}{a_r-k_{j}}\right)^{n_r}\eta_{j}^{(0)},
  ~~
  \quad
  A_{i j}=e^{a_{ij}}=\frac{\left(k_{i}-k_{j}\right)\left(q_{i}-q_{j}\right)}
  {\left(k_{i}-q_{j}\right)\left(q_{i}-k_{j}\right)},
  ~~~\sigma=3,4,5,
\end{equation}
where   the summation over $\mu$ means to take all possible $\mu_j= 0, 1~ (j =1, 2,\cdots,N)$,
or in Casoratian form
\begin{equation}\label{caso-solution}
  \tau=\left\lvert \widehat{N-1} \right\rvert_{[\lambda]}, ~~\lambda=n_1,n_2,\cdots,n_{\sigma}, ~~\sigma=3,4,5
\end{equation}
composed by $\psi=(\psi_1,\psi_2,\cdots,\psi_N)^T$ with  entries
\begin{equation}\label{psi-dAKP}
\psi_j(n_1,n_2,\cdots,n_{\sigma},;k_j,q_j)=\rho_j^+\prod_{r=1}^{\sigma}(a_r-q_j)^{n_r}
+\rho_j^-\prod_{r=1}^{\sigma}(a_r-k_j)^{n_r},~~~~\sigma=3,4,5.
\end{equation}

To approach to the bilinear ABS equations, we impose constraints
\begin{equation}\label{constraint condition}
q_j=-k_j
\end{equation}
for $j=1,2,\cdots,N$,
which yield reflection symmetry for the above $\tau$-function \cite{LNSZ-2021}:
\begin{subequations}\label{symmetries-1}
  \begin{align}
    &\tau((n_1,a_1),n_2,\cdots,n_{\sigma})=\tau((-n_1,-a_1),n_2,\cdots,n_{\sigma})=S_1 \tau,\\
    &\tau(n_1,(n_2,a_2),\cdots,n_{\sigma})=\tau(n_1,(-n_2,-a_2),\cdots,n_{\sigma})=S_2 \tau,\\
    &~~~~~~~~~~~~~~~~~~~~~~~~~~~~~~\vdots\nonumber\\
    &\tau(n_1,n_2,\cdots,(n_{\sigma},a_{\sigma}))=\tau(n_1,n_2,\cdots,(-n_{\sigma},-a_{\sigma}))=S_{\sigma} \tau,
    \end{align}
\end{subequations}
for $\sigma=3,4,5$,
where $S_i$ stands for the operator to reflect $(n_1,a_1)$ to $(-n_1,-a_1)$.
Such symmetries then lead us to the reflected forms of the original equations\footnote{
Note that the symmetries can be extended to more general cases, see Remark 2 in \cite{LNSZ-2021}}.
For example, when $\sigma=3$ and $\sigma=4$, the explicit reflected forms of the dAKP eqauation \eqref{3D-dAKP}
and 4D dAKP equation \eqref{4D-dAKP} are respectively
\begin{subequations}\label{symmetric-3D-dAKP}
\begin{align}
&A_1\equiv(a_2-a_3)\tau\tau_{123}+(a_3+a_1)\tau_3\tau_{12}-(a_1+a_2)\tau_2\tau_{13}=0,\label{be-3D-AKP-1}\\
&A_2\equiv(a_3-a_1)\tau\tau_{123}+(a_1+a_2)\tau_1\tau_{23}-(a_2+a_3)\tau_3\tau_{12}=0,\label{be-3D-AKP-2}\\
&A_3\equiv(a_1-a_2)\tau\tau_{123}+(a_2+a_3)\tau_2\tau_{13}-(a_3+a_1)\tau_1\tau_{23}=0,\label{be-3D-AKP-3}
\end{align}
\end{subequations}
and
\begin{subequations}\label{symmetric-4D-dAKP}
\begin{align}
  \mathfrak{A}_1\equiv & (a_2-a_3)(a_3-a_4)(a_2-a_4)\tau\tau_{1234}
  -(a_1+a_3)(a_3-a_4)(a_1+a_4)\tau_{12}\tau_{34}\nonumber\\
  &+(a_1+a_2)(a_2-a_4)(a_1+a_4)\tau_{13}\tau_{24}-(a_1+a_2)(a_2-a_3)(a_1+a_3)\tau_{14}\tau_{23}=0,
  \label{symmetric-4D-dAKP-a}\\
  \mathfrak{A}_2\equiv & (a_2+a_3)(a_3-a_4)(a_2+a_4)\tau_{12}\tau_{34}
  -(a_1-a_3)(a_3-a_4)(a_1-a_4)\tau\tau_{1234}\nonumber\\
  &-(a_1+a_2)(a_2+a_4)(a_1-a_4)\tau_{23}\tau_{14}
  +(a_1+a_2)(a_2+a_3)(a_1-a_3)\tau_{24}\tau_{13}=0,\\
  \mathfrak{A}_3\equiv & -(a_2+a_3)(a_3+a_4)(a_2-a_4)\tau_{13}\tau_{24}
  +(a_1+a_3)(a_3+a_4)(a_1-a_4)\tau_{23}\tau_{14}\nonumber\\
  &+(a_1-a_2)(a_2-a_4)(a_1-a_4)\tau\tau_{1234}-(a_1-a_2)(a_2+a_3)(a_1+a_3)\tau_{34}\tau_{12}=0,\\
  \mathfrak{A}_4\equiv & (a_2-a_3)(a_3+a_4)(a_2+a_4)\tau_{14}\tau_{23}
  -(a_1-a_3)(a_3+a_4)(a_1+a_4)\tau_{24}\tau_{13}\nonumber\\
  &+(a_1-a_2)(a_2+a_4)(a_1+a_4)\tau_{34}\tau_{12}-(a_1-a_2)(a_2-a_3)(a_1-a_3)\tau\tau_{1234}=0.
\end{align}
\end{subequations}
There are also reflected forms for the  5D dAKP equation.
Thus, the following results can be obtained.

\begin{theorem}\label{th-6}
The so-called symmetric dAKP system, consisting of the dAKP equation \eqref{3D-dAKP},
4D dAKP equation \eqref{4D-dAKP},
5D dAKP equation and their reflected forms,
share  the same soliton solutions in Hirota's exponential-polynomial form
\begin{equation}\label{tau-solution-Symmetric}
  \tau=\sum_{\mu=0,1} \exp \left[\sum_{j=1}^{N} \mu_{j} \eta_{j}
  +\sum_{1 \leqslant i<j}^{N} \mu_{i} \mu_{j} a_{i j}\right]
\end{equation}
with PWF and phase factor
\begin{equation}\label{PWF-symmetric-dAKP}
  \rho=\mathrm{e}^{\eta_{j}}=\prod_{r=1}^{\sigma} \left(\frac{a_r+k_{j}}{a_r-k_{j}}\right)^{n_r}\eta_{j}^{(0)},
  ~~\quad
  A_{i j}=e^{a_{i j}}=\frac{\left(k_{i}-k_{j}\right)^2}{\left(k_{i}+k_{j}\right)^2},~~~\sigma=3,4,5,
\end{equation}
or in Casoratian form
\begin{equation}\label{caso-solution-Symmetric}
  \tau=\big\lvert \widehat{N-1} \big\rvert_{[\lambda]}, ~~ \lambda=n_1,n_2,\cdots,n_{\sigma}~, ~~\sigma=3,4,5
\end{equation}
composed by $\psi=(\psi_1,\psi_2,\cdots,\psi_N)^T$ with  entries
\begin{equation}\label{psi-symmetric-dAKP}
\psi_j(n_1,n_2,\cdots,n_{\sigma},;k_j)=\rho_j^+\prod_{r=1}^{\sigma}(a_r+k_j)^{n_r}
+\rho_j^-\prod_{r=1}^{\sigma}(a_r-k_j)^{n_r},~~~~\sigma=3,4,5.
\end{equation}
\end{theorem}

\begin{remark}\label{Rem-1}
  Under the constraint condition \eqref{constraint condition},
  besides the symmetries \eqref{symmetries-1},
  the $\tau$-function allows some other symmetries as well.
For example, when $\sigma=4$,
 there is a symmetry
 \[\tau=S_2\circ S_1 \tau,\]
which, from the 4D AKP equation \eqref{4D-dAKP}, leads to a reflected 4D dAKP equation
\begin{subequations}\label{symmetric-4D-dAKP-2}
\begin{align}
  \mathfrak{A}_5\equiv & (a_2+a_3)(a_3-a_4)(a_2+a_4)\tau_2\tau_{134}
  -(a_1+a_3)(a_3-a_4)(a_1+a_4)\tau_{1}\tau_{234}\nonumber\\
  &-(a_1-a_2)(a_2+a_4)(a_1+a_4)\tau_{123}\tau_{4}+(a_1-a_2)(a_2+a_3)(a_1+a_3)\tau_{124}\tau_{3}=0.
\end{align}
There are also symmetries
  \[
  \tau=S_3\circ S_1 \tau,~~ \tau=S_4\circ S_1 \tau,
  \]
which give rise to the following reflected 4D dAKP equations, respectively,
\begin{align}
 \mathfrak{A}_6\equiv & -(a_2+a_3)(a_3+a_4)(a_2-a_4)\tau_3\tau_{124}
  +(a_1-a_3)(a_3+a_4)(a_1+a_4)\tau_{123}\tau_{4}\nonumber\\
  &+(a_1+a_2)(a_2-a_4)(a_1+a_4)\tau_{1}\tau_{234}-(a_1+a_2)(a_2+a_3)(a_1-a_3)\tau_{134}\tau_{2}=0,\\
  \mathfrak{A}_7\equiv &(a_2-a_3)(a_3+a_4)(a_2+a_4)\tau_4\tau_{123}
  -(a_1+a_3)(a_3+a_4)(a_1-a_4)\tau_{124}\tau_{3}\nonumber\\
  &+(a_1+a_2)(a_2+a_4)(a_1-a_4)\tau_{134}\tau_{2}-(a_1+a_2)(a_2-a_3)(a_1+a_3)\tau_{1}\tau_{234}=0,
\end{align}
\end{subequations}
and these equations allow the same soliton solution as those given in Theorem \ref{th-6}.
Equations in \eqref{symmetric-4D-dAKP-2} also belong to the symmetric dAKP system
which will generate bilinear ABS equations in reduction.
\end{remark}

\section{Reductions to the bilinear ABS}\label{sec-5}

In this section, we explain how the bilinear ABS equations are connected to
the symmetric dAKP system via  reductions or continuum limits.

\subsection{Bilinear H1}\label{sec-5-1}

Equation \eqref{be-f-H1a} in the bilinear H1 \eqref{be-f-H1}
is nothing but the bilinear differential-difference AKP equation with one continuous independent variable
(see case (N-2) in \cite{DJM-1982-II}), which is a result of a
continuum limit of the dAKP equation \eqref{3D-dAKP}.
In fact,
let
\begin{equation}\label{continuum limit condition-1}
  a_3,n_3\rightarrow \infty ~~~\text{while keeping}~ x=n_3/a_3 ~~\text{finite},
\end{equation}
and still write $\tau(n_1,n_2,n_3)$ as $\tau(n_1,n_2,x)$ without making confusion, i.e.,
\begin{equation}
 \tau= \tau(n_1,n_2,n_3):=\tau(n_1,n_2,x),
\end{equation}
it then implies that
\begin{equation}
  \tau_3=\tau(n_1,n_2,n_3+1)= \tau\Bigl(n_1,n_2,x+\frac{1}{a_3}\Bigr)
=\tau+\frac{1}{a_3}\partial_x\tau+\frac{1}{2a_3^2}\partial_x^2\tau+\cdots.
\end{equation}
Inserting this Taylor expansion into equation \eqref{3D-dAKP}, from the leading term
we get
\begin{subequations}\label{semi-d-AKP-1}
\begin{equation}
    (a_1-a_2)(\tau\tau_{12}-\tau_1\tau_2)+D_x\tau_1\cdot\tau_2=0, \label{semi-d-AKP-1-1}
\end{equation}
which is \eqref{be-f-H1a}.
Starting from the reflected dAKP equation \eqref{be-3D-AKP-2}
and implementing the same continuum limit yields
\begin{equation}
    (a_1+a_2)(\tau\tau_{12}-\tau_1\tau_2)+D_x\tau\cdot\tau_{12}=0,\label{semi-d-AKP-1-2}
\end{equation}
i.e., \eqref{be-f-H1b}.
\end{subequations}
Thus, the bilinear H1 \eqref{be-f-H1} consists of two differential-difference AKP equations,
which are the continuum limits of the dAKP equation \eqref{3D-dAKP}
and its reflected form \eqref{be-3D-AKP-2}.

As for the $\tau$-function of \eqref{semi-d-AKP-1}
in the form \eqref{tau-solution-Symmetric} or \eqref{caso-solution-Symmetric},
considering the continuum limit of   \eqref{PWF-symmetric-dAKP} and   \eqref{psi-symmetric-dAKP} 
with $\sigma=3$, we have
\begin{equation}\label{PWF-semi-dAKP-1}
 \rho=e^{2k_jx}\left(\frac{a_1+k_{j}}{a_1-k_{j}}\right)^{n_1}\left(\frac{a_2+k_{j}}{a_2-k_{j}}\right)^{n_2}
\eta_{j}^{(0)}
\end{equation}
and
\begin{align}\label{psi-H1-1}
  \psi_j(n_1,n_2,x;k_j)&=\rho_j^+e^{k_jx}(a_1+k_j)^{n_1}(a_2+k_j)^{n_2}
  +\rho_j^-e^{-k_jx}(a_1-k_j)^{n_1}(a_2-k_j)^{n_2}.
\end{align}

\subsection{Bilinear H2}\label{sec-5-2}

We start from the following two reflected equations in the dAKP system:
\begin{subequations}\label{3d-dAKP-134}
  \begin{align}
  & (a_3-a_1)\tau\tau_{134}+(a_1+a_4)\tau_1\tau_{34}-(a_4+a_3)\tau_3\tau_{14}=0,\\
  &(a_3-a_2)\tau\tau_{234}+(a_2+a_4)\tau_2\tau_{34}-(a_4+a_3)\tau_3\tau_{24}=0,\label{3d-dAKP-234}
    \end{align}
\end{subequations}
where their common PWF is
\begin{equation}\label{PWF-3D-symmetric-dAKP-1234}
  \rho=\eta_{j}^{(0)}\prod^4_{i=1}\left(\frac{a_i+k_j}{a_i-k_{j}}\right)^{n_i} ,
 \end{equation}
and the common Casoratian entries are
\begin{align}\label{psi-H2-13}
  \psi_j(n_1,n_2,n_3,n_4;k_j)&=\rho_j^+ \prod^4_{l=1}(a_l+k_j)^{n_l}
  +\rho_j^- \prod^4_{l=1}(a_l-k_j)^{n_l}.
\end{align}
Impose a constraint on the  parameters and introduce a new variable $N_3$ as\footnote{
The constraint $a_3=a_4$ corresponds to the constraint on $\tau$:
$\tau(n_1,n_2,n_3+1,n_4)=\tau(n_1,n_2,n_3,n_4+1)$.}
\begin{equation}
a_3=a_4,~~N_3=a_3+a_4,
\end{equation}
which turn the PWF \eqref{PWF-3D-symmetric-dAKP-1234} to a form
\begin{equation}\label{PWF-2D-symmetric-dAKP-123}
  \rho=\left(\frac{a_1+k_j}{a_1-k_{j}}\right)^{n_1}\left(\frac{a_2+k_j}{a_2-k_{j}}\right)^{n_2}
  \left(\frac{a_3+k_j}{a_3-k_{j}}\right)^{N_3},
 \end{equation}
then the system \eqref{3d-dAKP-134}  is reduced  to
\begin{subequations}\label{2D-dAKP-123}
\begin{align}
  & (a_3-a_1)\tau\tau_{133}+(a_1+a_3)\tau_1\tau_{33}-2a_3\tau_3\tau_{13}=0,\\
  &(a_3-a_2)\tau\tau_{233}+(a_2+a_3)\tau_2\tau_{33}-2a_3\tau_3\tau_{23}=0.
\end{align}
\end{subequations}
Note that here $\tau_3$ means $\tau(N_3+1)$.
Next, let
\begin{equation}\label{continuum limit condition-2}
  a_3,N_3\rightarrow \infty ~~~\text{while keeping}~ x=N_3/a_3 ~~\text{finite},
\end{equation}
and rewrite the dependent variable as
\begin{equation}
 \tau= \tau(n_1,n_2,N_3):=\tau(n_1,n_2,x).
\end{equation}
Then, one can expand \eqref{2D-dAKP-123} in terms of $1/a_3$
and from the leading terms we have
the bilinear equations \eqref{be-f-H2} of H2, i.e.,
\begin{subequations}\label{be-H2-f-2}
  \begin{align}
    &  (D_x^2+2a_1D_x) \tau\cdot  \tau_1=0,\\
    &  (D_x^2+2a_2D_x) \tau\cdot  \tau_2=0.
  \end{align}
\end{subequations}
The corresponding PWF $\rho$ and  the Casoratian entries $\psi_j$
are given as
\begin{equation}\label{PWF-semi-dAKP-H2-1}
  \rho=\left(\frac{a_1+k_{j}}{a_1-k_{j}}\right)^{n_1}
  \left(\frac{a_2+k_{j}}{a_2-k_{j}}\right)^{n_2} e^{2k_j x}\eta_{j}^{(0)},
\end{equation}
and
\begin{align}\label{psi-H2-13s}
  \psi_j(n_1,n_2,x;k_j)&=\rho_j^+e^{k_j x}
  \prod^2_{l=1}(a_l+k_j)^{n_l} +\rho_j^-e^{-k_j x}\prod^2_{l=1}(a_l-k_j)^{n_l}.
\end{align}

\subsection{Bilinear H3}\label{sec-5-3}

As mentioned in subsection \ref{sec-3-3}, there are two bilinear forms of H3 equation,
i.e., equations \eqref{be-H3-f-1} and \eqref{be-H3-f-2}.
In equation \eqref{be-H3-f-1}, the first two equations can be considered as
the reflected dAKP equations, which are all featured by their total index (1,1,1).
However, the third equation  \eqref{be-H3-3} has total index (1,1,0).
In the next part we will look at the 4D dAKP equation  which has total index (1,1,1,1)
and try obtaining  \eqref{be-H3-3} through dimensional reductions.

Consider the reflected 4D dAKP equation \eqref{symmetric-4D-dAKP-a},
which has a 4D PWF \eqref{PWF-3D-symmetric-dAKP-1234}.
Introduce constraint $a_4=-a_3$, which gives rise to the reduction
\begin{equation}\label{reduction-H3-tau-2}
\tau(n_1,n_2,n_3+1,n_4)=\tau(n_1,n_2,n_3,n_4-1).
\end{equation}
The  PWF \eqref{PWF-3D-symmetric-dAKP-1234} is reduced to
\begin{equation}\label{PWF-H3}
  \rho=\left(\frac{a_1+k_j}{a_1-k_{j}}\right)^{n_1}\left(\frac{a_2+k_j}{a_2-k_{j}}\right)^{n_2}
  \left(\frac{a_3+k_j}{a_3-k_{j}}\right)^{N_3}\eta_{j}^{(0)},
 \end{equation}
 where $N_3 =n_3-n_4$,
 and the  reflected 4D dAKP equation \eqref{symmetric-4D-dAKP-a}
gives rise to
\begin{align}
2a_3(a_1-a_2)\tau\tau_{12}+(a_3-a_1)(a_2+a_3)\tau_{13}\tau_2^3+(a_1+a_3)(a_2-a_3)\tau_{23}\tau_1^3=0,
\label{3D-H3-3}
\end{align}
which is exactly the third bilinear equation \eqref{be-H3-3}.
Here by $\tau_3$ and $\tau^3$ we mean $\tau(N_3+1)$ and $\tau(N_3-1)$, respectively.
With the above PWF \eqref{PWF-H3}, the reflected dAKP equations \eqref{be-3D-AKP-2} and \eqref{be-3D-AKP-1}
hold, and so do \eqref{be-H3-1} and \eqref{be-H3-2}.
For convenience, we replace $N_3$ by $n_3$ in \eqref{PWF-H3}
and then the $\tau$-function in \eqref{be-H3-f-1} will be given by Theorem \ref{th-6} with $\sigma=3$.

The second bilinear form \eqref{be-H3-f-2} consists of two  2D equations in 3D space.
The total index patterns of them are $(1,0,0)$ and $(0,1,0)$, respectively.
Note that $n_2$ is a dummy variable in \eqref{be-H3-1-1},
and  likewise  $n_1$ for \eqref{be-H3-1-2}.
To get these two equations simultaneously, we consider two 3D dAKP equations in 4D space:
\begin{subequations}\label{3D-dAKP-two}
 \begin{align}
 & (a_1-a_3)\tau_4 \tau_{13}+(a_3-a_4) \tau_1\tau_{34}+(a_4-a_1) \tau_3 \tau_{14}=0,\\
 & (a_2-a_3)\tau_4 \tau_{23}+(a_3-a_4) \tau_2\tau_{34}+(a_4-a_2) \tau_3 \tau_{24}=0,
 \end{align}
 \end{subequations}
which allow a $\tau$-function \eqref{tau-solution-Symmetric}
with the PWF \eqref{PWF-3D-symmetric-dAKP-1234}\footnote{Note that
both equations in \eqref{3D-dAKP-two} are dAKP equations, which allow
a more general PWF like the one in \eqref{PWF-dAKP} or in \eqref{psi-dAKP} with $\sigma=4$.
However, to meet the definition $N_3=n_3-n_4$, we shall consider the dKdV type PWFs
in \eqref{PWF-symmetric-dAKP} or \eqref{psi-symmetric-dAKP} with $\sigma=4$.}.
Introducing reduction $a_4=-a_3$  which corresponds to the reduction \eqref{reduction-H3-tau-2} on $\tau$,
the PWF \eqref{PWF-3D-symmetric-dAKP-1234} goes to \eqref{PWF-H3} with $N_3 =n_3-n_4$
and \eqref{3D-dAKP-two} gives rise to \eqref{be-H3-f-2}.

\subsection{Bilinear Q1}\label{sec-5-4}

For  Q1 equation with a power background,
its bilinear form \eqref{be-f-Q1-1} is the same as  \eqref{be-H3-f-2} of H3 equation,
which can be obtained from the dAKP equations \eqref{3D-dAKP-two} through a one-step reduction.
In the next part we look at the bilinear form \eqref{be-f-Q1-2} for Q1 equation with a linear background,
which is related to skew-limit and much more complicated than the power background case.

The first three equations in \eqref{be-f-Q1-2} are the dAKP and its reflected forms,
while \eqref{be-Q1-3} and \eqref{be-Q1-4}
are (3+1)D bilinear equations which have total index $(1,1,1)$ in discrete directions.
To obtain these two equations, we take a two-step procedure.
We will start from certain 5D dAKP system with total index $(1,1,1,1,1)$
and take dimensional reductions to  obtain some 4D equations with total index $(1,1,1,0)$,
on which we then perform a skew continuum limit to get some (3+1)D semi-discrete equations
with total index $(1,1,1)$ in discrete directions.

Let us consider the 5D dAKP equation \eqref{3-4-5-D-AKP}, i.e.,
\begin{equation}\label{5-D-AKP}
  \mathcal{A}  \equiv\left | \begin{matrix}
     1 &a_1 & a_1^2 & a_1^{3}  &\tau_1 \tau_{2345} \\
     1 &a_2 & a_2^2  &a_2^{3}&\tau_2\tau_{1345}\\
     1 &a_3 & a_3^2  &a_3^{3}&\tau_3\tau_{1245}\\
     1 &a_4 & a_4^2  &a_4^{3}&\tau_4\tau_{1235}\\
     1 &a_5&a_5^2  &a_5^{3} &\tau_5\tau_{1234}\\
     \end{matrix} \right |=0,
   \end{equation}
and its 5 reflected forms
\begin{subequations}\label{5D-ref-AKP}
  \begin{align}
    & \mathcal{A}_1= \mathcal{A}|_{a_1\rightarrow-a_1,n_1\rightarrow-n_1},\\
    & \mathcal{A}_2= \mathcal{A}|_{a_2\rightarrow-a_2,n_2\rightarrow-n_2},\\
    & \mathcal{A}_3= \mathcal{A}|_{a_3\rightarrow-a_3,n_3\rightarrow-n_3},\\
    & \mathcal{A}_4= \mathcal{A}|_{a_4\rightarrow-a_4,n_4\rightarrow-n_4},\\
    & \mathcal{A}_5= \mathcal{A}|_{a_5\rightarrow-a_5,n_5\rightarrow-n_5},
  \end{align}
\end{subequations}
with the PWF
\begin{equation}\label{PWF-Q1-1}
  \rho=\eta_{j}^{(0)}\prod^5_{i=1}\left(\frac{a_i+k_{j}}{a_i-k_{j}}\right)^{n_i}.
  \end{equation}
Next, by imposing constraint
\begin{equation}\label{constraint-Q1}
  a_4=-a_3,
\end{equation}
the PWF \eqref{PWF-Q1-1} is reduced to
 \begin{equation}\label{PWF-Q1-3}
  \rho=\left(\frac{a_1+k_{j}}{a_1-k_{j}}\right)^{n_1}\left(\frac{a_2+k_{j}}{a_2-k_{j}}\right)^{n_2}
  \left(\frac{a_3+k_{j}}{a_3-k_{j}}\right)^{N_3}\left(\frac{a_5+k_{j}}{a_5-k_{j}}\right)^{n_5}\eta_{j}^{(0)}
   \end{equation}
where $N_3=n_3-n_4$.
The constraint \eqref{constraint-Q1} indicates
$\tau(n_1,n_2,n_3+1,n_4,n_5)=\tau(n_1,n_2,n_3,n_4-1,n_5),$
which leads the 5D system \eqref{5-D-AKP} and \eqref{5D-ref-AKP}
to the following 4D system with total index $(1,1,0,1)$:
\begin{subequations}\label{4D-AKP-equations-1}
\begin{align}
  \mathcal{A}_1^{(1)}\equiv&- (a_1 - a_2) (a_1 - a_3) (a_1 - a_5) (a_2 - a_3) (a_2 - a_5) (a_3 - a_5)
  \tau_{1235} \tau^3\nonumber \\
  &- (a_1 - a_2) (a_1 + a_3) (a_1 - a_5) (a_2 + a_3) (a_2 - a_5) (a_3 + a_5) \tau_{125}^3 \tau_3\nonumber\\
&+2 a_3(a_1 - a_2) (a_1^2 - a_3^2) (a_2^2 - a_3^2)  \tau_{12} \tau_5
 + 2 a_3  (a_1 - a_5) (a_1^2 - a_3^2) (a_3^2 - a_5^2)\tau_{15} \tau_2 \nonumber\\
 &- 2 a_3(a_2 - a_5) (a_2^2 - a_3^2) (a_3^2 - a_5^2) \tau_1\tau_{25}=0,\label{A-1}\displaybreak\\
 \mathcal{A}_2^{(1)}\equiv&  (a_1 +a_2) (a_2 - a_3) (a_1 + a_3) (a_2 - a_5) (a_3 - a_5) (a_1 + a_5)
 \tau_1^3\tau_{235}  \nonumber\\
 & + (a_1 + a_2) (a_1 - a_3) (a_2 + a_3) (a_2 - a_5) (a_1 +  a_5) (a_3 + a_5) \tau_{13} \tau_{25}^3\nonumber\\
&-2 a_3 (a_2 - a_5)(a_2^2 - a_3^2) (a_3^2 - a_5^2) \tau \tau_{125}
- 2a_3 (a_1 + a_2) (a_1^2 - a_3^2) (a_2^2 - a_3^2)  \tau_{15} \tau_2\nonumber\\
 &  - 2a_3  (a_1 + a_5)  (a_1^2 - a_3^2) (a_3 ^2- a_5^2) \tau_{12}\tau_5=0,\label{A1-1}\\
 \mathcal{A}_3^{(1)}\equiv &- (a_1 + a_2) (a_1 - a_3) (a_2 + a_3) (a_1 - a_5) (a_3 - a_5) (a_2 + a_5)
 \tau_{135} \tau_2^3\nonumber\\
 &- (a_1 + a_2) (a_2 - a_3) (a_1 + a_3) (a_1 - a_5) (a_2 + a_5) (a_3 + a_5) \tau_{15}^3 \tau_{23} \nonumber\\
 &+ 2 a_3 (a_2 + a_5) (a_2^2 - a_3^2) (a_3^2 - a_5^2)  \tau_{12} \tau_5
 +2 a_3 (a_1 - a_5) (a_1^2 -a_3^2) (a_3^2 - a_5^2)  \tau \tau_{125} \nonumber\\
 &+ 2  a_3(a_1 + a_2) (a_1^2 - a_3^2) (a_2^2 - a_3^2) \tau_1 \tau_{25}=0 ,\label{A2-1}\\
 \mathcal{A}_4^{(1)}\equiv& -(a_1 - a_2) (a_1 + a_3) (a_2 + a_3) (a_1 + a_5) (a_2 + a_5) (a_3 - a_5)
 \tau_{12}^3 \tau_{35}  \nonumber\\
& -(a_1 - a_2) (a_1 - a_3) (a_2 - a_3) (a_1 +  a_5) (a_2 + a_5) (a_3 + a_5)\tau_{123} \tau_5^3\nonumber\\
 &+2 a_3(a_1 - a_2) (a_1^2 - a_3^2) (a_2^2 - a_3^2)  \tau \tau_{125}
   - 2a_3 (a_2 + a_5)(a_2^2 - a_3^2) (a_3^2 - a_5^2) \tau_{15} \tau_2\nonumber\\
  & + 2 a_3 (a_1 + a_5) (a_1^2 - a_3^2) (a_3^2- a_5^2) \tau_1 \tau_{25}=0.\label{A3-1}
\end{align}
\end{subequations}
Note that here $\tau_3=\tau(N_3+1)$.
To perform skew continuum limit, we need to convert the above system
from the coordinates $\{n_1,n_2,N_3,n_5\}$ into  the coordinates $\{n_1,n_2,N_3^\prime=N_3+n_5,n_5\}$.
After that, we come to a 4D system with total index $(1,1,1,1)$:
\begin{subequations}\label{4D-AKP-equations-2}
  \begin{align}
    \mathcal{A}_1^{(2)}\equiv&- (a_1 - a_2) (a_1 - a_3) (a_1 - a_5) (a_2 - a_3) (a_2 - a_5) (a_3 - a_5)
    \tau_{12335}  \tau^3\nonumber \\
    &- (a_1 - a_2) (a_1 + a_3) (a_1 - a_5) (a_2 + a_3) (a_2 - a_5) (a_3 + a_5) \tau_{125} \tau_3\nonumber\\
  &+2 a_3(a_1 - a_2) (a_1^2 - a_3^2) (a_2^2 - a_3^2)  \tau_{12} \tau_{35}
   + 2 a_3  (a_1 - a_5) (a_1^2 - a_3^2) (a_3^2 - a_5^2)\tau_{135} \tau_2 \nonumber\\
   &- 2 a_3(a_2 - a_5) (a_2^2 - a_3^2) (a_3^2 - a_5^2)  \tau_1\tau_{235}=0,\\
   \mathcal{A}_2^{(2)}\equiv&  (a_1 +a_2) (a_2 - a_3) (a_1 + a_3) (a_2 - a_5) (a_3 - a_5) (a_1 + a_5)
    \tau_1^3\tau_{2335}  \nonumber\\
   & + (a_1 + a_2) (a_1 - a_3) (a_2 + a_3) (a_2 - a_5) (a_1 +  a_5) (a_3 + a_5) \tau_{13} \tau_{25}\nonumber\\
  &-2 a_3 (a_2 - a_5)(a_2^2 - a_3^2) (a_3^2 - a_5^2) \tau \tau_{1235}
  - 2a_3 (a_1 + a_2) (a_1^2 - a_3^2) (a_2^2 - a_3^2)  \tau_{135} \tau_2\nonumber\\
   &  - 2a_3  (a_1 + a_5)  (a_1^2 - a_3^2) (a_3 ^2- a_5^2) \tau_{12}\tau_{35}=0,\\
   \mathcal{A}_3^{(2)}\equiv &- (a_1 + a_2) (a_1 - a_3) (a_2 + a_3) (a_1 - a_5) (a_3 - a_5) (a_2 + a_5)
   \tau_{1335} \tau_2^3\nonumber\\
   &- (a_1 + a_2) (a_2 - a_3) (a_1 + a_3) (a_1 - a_5) (a_2 + a_5) (a_3 + a_5) \tau_{15} \tau_{23} \nonumber\\
   &+ 2 a_3 (a_2 + a_5) (a_2^2 - a_3^2) (a_3^2 - a_5^2)  \tau_{12} \tau_{35}
   +2 a_3 (a_1 - a_5) (a_1^2 -a_3^2) (a_3^2 - a_5^2)  \tau \tau_{1235} \nonumber\\
   &+ 2  a_3(a_1 + a_2) (a_1^2 - a_3^2) (a_2^2 - a_3^2) \tau_1 \tau_{235}=0 , \\
   \mathcal{A}_4^{(2)}\equiv& -(a_1 - a_2) (a_1 + a_3) (a_2 + a_3) (a_1 + a_5) (a_2 + a_5) (a_3 - a_5)
   \tau_{12}^3 \tau_{335}  \nonumber\\
  & -(a_1 - a_2) (a_1 - a_3) (a_2 - a_3) (a_1 +  a_5) (a_2 + a_5) (a_3 + a_5) \tau_{123} \tau_5\nonumber\\
   &+2 a_3(a_1 - a_2) (a_1^2 - a_3^2) (a_2^2 - a_3^2)  \tau \tau_{1235}
    - 2a_3 (a_2 + a_5)(a_2^2 - a_3^2) (a_3^2 - a_5^2) \tau_{135} \tau_2\nonumber\\
    & + 2 a_3 (a_1 + a_5) (a_1^2 - a_3^2) (a_3^2- a_5^2) \tau_1 \tau_{235}=0,
  \end{align}
  \end{subequations}
and correspondingly, the PWF takes a form
\begin{equation}\label{PWF-4D-dAKP-3}
  \rho=\left(\frac{a_1+k_{j}}{a_1-k_{j}}\right)^{n_1}\left(\frac{a_2+k_{j}}{a_2-k_{j}}\right)^{n_2}
  \left(\frac{a_3+k_{j}}{a_3-k_{j}}\right)^{N_3^\prime}
  \left(\frac{a_3-k_{j}}{a_3+k_{j}}\cdot
  \frac{a_5+k_{j}}{a_5-k_{j}}\right)^{n_5}\eta_{j}^{(0)}.
  \end{equation}
Note that in \eqref{4D-AKP-equations-2} by $\tau_3$ we mean $\tau(N_3^\prime+1)$.
Now set $\varepsilon=a_5-a_3$.
Then the limit scheme
\begin{equation}\label{limit}
  N_3\rightarrow-\infty,~n_5\rightarrow\infty,~\varepsilon\rightarrow 0,~~
  \text{while}~ n_3 ~\text{and}~ y:= n_5\varepsilon~\text{being finite},
\end{equation}
brings the factor $ \left(\frac{a_3-k_{j}}{a_3+k_{j}}\cdot
  \frac{a_5+k_{j}}{a_5-k_{j}}\right)^{n_5}$ in the PWF \eqref{PWF-4D-dAKP-3} to
\begin{align*}
&\underset{\varepsilon\to 0}{\underset{n_5\to \infty}{\lim}}
\left(1+\frac{2k_j(a_3-a_5)n_5}{(a_3+k_j)(a_5-k_j)n_5}\right)^{n_5}
=\underset{\varepsilon\to 0}{\underset{n_5\to \infty}{\lim}}
\left(1+\frac{-2k_j(y-y_0)}{(a_3+k_j)(a_3-k_j+\varepsilon)n_5}\right)^{n_5} \\
=&\underset{\varepsilon\to 0}{\underset{n_5\to \infty}{\lim}}
\left(1+\frac{1}{n_5}\cdot\frac{-2k_j(y-y_0)}{(a_3^2-k_j^2)+\varepsilon(a_3+k_j)}\right)^{n_5}
={\underset{n_5\to \infty}{\lim}}\left(1+\frac{1}{n_5}\cdot
\frac{-2k_j(y-y_0)}{(a_3^2-k_j^2)}\right)^{n_5} \\
=&\exp\left(\frac{-2k_jy}{a_3^2-k_j^2}\right).
\end{align*}
For convenience, we write
\begin{equation}
  \tau=\tau(n_1,n_2,N_3^\prime,n_5):=\tau(n_1,n_2,N_3^\prime,y),
\end{equation}
which indicates
\begin{equation}
  \tau_5=\tau(n_1,n_2,N_3^\prime,n_5+1)=\tau(n_1,n_2,N_3^\prime,y+\varepsilon)
  =\tau+\varepsilon \partial_y\tau+\frac{1}{2}\varepsilon^2\partial_y^2\tau+\cdots.
\end{equation}
In such a scheme, the leading terms of the  system \eqref{4D-AKP-equations-2} yield
\begin{subequations}\label{(3+1)AKP}
  \begin{align}
    \mathcal{A} _1^{(3)}\equiv&-4a_3^2 (a_1 + a_3) (a_1 - a_3)^2  \tau_{13} \tau_2
    + (a_1 - a_2) (a_1 -  a_3)^2 (a_2 - a_3)^2 \tau_{1233} \tau^3\nonumber\\
   & -(a_1 - a_2) (a_1 + a_3) (a_2 +  a_3) (a_1 (a_2 - 3 a_3) + a_3 (-3 a_2 + 5 a_3)) \tau_{12} \tau_3 \nonumber\\
    & +4  a_3^2(a_2 + a_3)(a_2 - a_3)^2 \tau_1 \tau_{23}
    + 2 a_3(a_1 - a_2)  (a_1^2 - a_3^2) (a_2^2 - a_3^2)D_y\tau_3\!\cdot\!\tau_{12}=0 ,\\
    \mathcal{A} _2^{(3)}\equiv&4 a_3^2  (a_2 + a_3)(a_2 - a_3)^2 \tau \tau_{123}
    - (a_1 + a_2) (a_2 - a_3)^2 (a_1 + a_3)^2 \tau_1^3 \tau_{233}\nonumber\\
& + (a_1 + a_2) (a_1 - a_3) (a_2 + a_3) (a_1 (a_2 - 3 a_3) + (3 a_2 - 5 a_3) a_3) \tau_{13} \tau_2\nonumber\\
&+ 4 a_3^2  (a_1 - a_3)(a_1 + a_3)^2 \tau_{12} \tau_3+2 a_3 (a_1 + a_2) (a_1^2 - a_3^2) (a_2^2- a_3^2)
D_y \tau_2\!\cdot\!\tau_{13} =0,\\
\mathcal{A} _3^{(3)}\equiv&-4 a_3^2 (a_1 + a_3) (a_1 - a_3)^2 \tau \tau_{123}
+(a_1 + a_2) (a_1 - a_3)^2 (a_2 + a_3)^2 \tau_{133}\tau_2^3\nonumber\\
&-(a_1 + a_2) (a_2 - a_3) (a_1 + a_3) (a_1 (a_2 + 3 a_3) - a_3 (3 a_2 + 5 a_3)) \tau_1\tau_{23}\nonumber\\
&-4 a_3^2  (a_2 - a_3)(a_2 + a_3)^2 \tau_{12}\tau_3- 2a_3  (a_1 + a_2) (a_1^2 - a_3^2) (a_2^2 - a_3^2)
 D_y\tau_1\!\cdot\!\tau_{23}=0, \\
\mathcal{A} _4^{(3)}\equiv&4 a_3^2 (a_2 - a_3) (a_2 + a_3)^2 \tau_{13}\tau_2
+ (a_1 - a_2) (a_1 + a_3)^2 (a_2 + a_3)^2 \tau_{12}^3\tau_{33}\nonumber\\
&-(a_1 - a_3) (a_2 - a_3) (a_1^2 (a_2 + 3 a_3) - a_2 a_3 (3 a_2 + 5 a_3) - a_1 (a_2^2 - 5 a_3^2))
 \tau\tau_{123}\nonumber\\
&-\!4 a_3^2 (a_1 - a_3) (a_1 + a_3)^2 \tau_1\tau_{23}\! -\! 2 a_3 (a_1 - a_2)(a_1^2 - a_3^2) (a_2^2 - a_3^2)
D_y\tau\!\cdot\!\tau_{123}=0.
  \end{align}
\end{subequations}
Thus, we arrive at a (3+1)D system  with total index $(1,1,1)$ in three discrete directions,
and the corresponding PWF  is
\begin{equation}\label{PWF-(3+1)D-dAKP}
  \rho= \left(\frac{a_1+k_{j}}{a_1-k_{j}}\right)^{n_1}\left(\frac{a_2+k_{j}}{a_2-k_{j}}\right)^{n_2}
  \left(\frac{a_3+k_{j}}{a_3-k_{j}}\right)^{N_3^\prime}\exp\left(\frac{-2k_jy}{a_3^2-k_j^2}\right)\eta_{j}^{(0)}.
   \end{equation}

To recover \eqref{be-Q1-3} and \eqref{be-Q1-4}, i.e., $\mathcal{Q}_3=0$ and $\mathcal{Q}_4=0$,
we consider combinations
   \begin{subequations}
    \begin{align}
      &(a_1+a_2) \mathcal{A}_1^{(3)}+(a_1- a_3) \mathcal{A}_2^{(3)}+(a_2-a_3) \mathcal{A}_3^{(3)}=0,\\
      &(a_2+a_3) \mathcal{A}_2^{(3)}+( a_1+a_3)\mathcal{A}_3^{(3)}-(a_1+a_2) \mathcal{A}_4^{(3)}=0,
    \end{align}
    \end{subequations}
    which give rise to
    \begin{subequations}\label{Q3-Q4-A}
      \begin{align}
        &\{2 a_1 a_3 (a_2^2 - a_3^2) +a_3^2 (7 a_2^2 - 2 a_2 a_3 - 13 a_3^2)
        +  a_1^2 (-a_2^2 + 2 a_2 a_3 + 7 a_3^2)\} A\nonumber\\
        &-(a_2-a_3)(a_1 - a_3) \dot{A}- 2 a_3 (a_1^2 - a_3^2) (a_2^2 - a_3^2)\mathcal{Q}_3=0,\label{Q3-A}\\
      &\{  2 a_1 a_3 (a_2^2 - a_3^2)+ a_3^2 (-7 a_2^2 - 2 a_2 a_3 + 13 a_3^2)
      +a_1^2 (a_2^2 + 2 a_2 a_3 - 7 a_3^2)\} A_3\nonumber\\
     & +(a_1 + a_3) (a_2 + a_3) \dot{A}_1+2 a_3 (a_1^2 - a_3^2) (a_2^2 - a_3^2)\mathcal{Q}_4=0,\label{Q4-A}
      \end{align}
      \end{subequations}
where $A$ and $A_3$ are the dAKP \eqref{3D-dAKP} and reflected form \eqref{be-3D-AKP-3}, respectively,
and the  notation $\dot{A}$ stands for
\begin{align}\label{dot A}
\dot{A}:=& (a_2 -a_3)(a_1 + a_3)^2\tau_1^3\tau_{233}- (a_1 - a_2) (a_1 - a_3) (a_2 - a_3)\tau^3\tau_{1233}
\nonumber\\
 & -(a_1 - a_3)(a_2 +a_3)^2  \tau_2^{3} \tau_{133} +4a_3^2  (a_1 - a_2) \tau\tau_{123},
\end{align}
and $\dot{A}_1$ is a reflected form of $\dot{A}$ \eqref{dot A} by taking
\begin{equation}
  \tau(n_1,n_2,(N_3^\prime,a_3))=\tau(n_1,n_2,(-N_3^\prime,-a_3)).
\end{equation}
Thus, if we can prove $\dot{A}=0$, we may recover equations $\mathcal{Q}_3=0$ and
$\mathcal{Q}_4=0$ from \eqref{Q3-Q4-A}.
Meanwhile, the $\tau$-function composed by the PWF \eqref{PWF-(3+1)D-dAKP}
satisfies the first three bilinear equations in \eqref{be-f-Q1-2}.

As for equation  $\dot{A}=0$, which is an  8-point 3D equation, i.e.,
\begin{align}\label{alter-dAKP}
  \dot{A}:=& (a_2 -a_3)(a_1 + a_3)^2\tau_1^3\tau_{233}- (a_1 - a_2) (a_1 - a_3) (a_2 - a_3)\tau^3\tau_{1233}
  \nonumber\\
  & -(a_1 - a_3)(a_2 +a_3)^2  \tau_2^3 \tau_{133} +4a_3^2  (a_1 - a_2) \tau\tau_{123}=0,
\end{align}
we have the following remarks. More details about this equation can be found in Appendix \ref{App-A}.
\begin{remark}\label{Rem-2}
Equation \eqref{alter-dAKP} holds when  $\tau$ is defined as in Theorem \ref{th-6} with $\sigma=4$.
A direct proof can be found in Appendix \ref{App-A}.
\end{remark}

\begin{remark}\label{Rem-3}
Performing the straight continuum limit
 \[ a_3,n_3\rightarrow \infty, ~~ x=n_3/a_3, ~~x~\text{keeps  finite},\]
and still writing $\tau(n_1,n_2,n_3)=\tau(n_1,n_2,x)$,
 the leading term of \eqref{alter-dAKP}  gives rise to the semi-discrete AKP equation \eqref{semi-d-AKP-1-1}, i.e.,
\begin{equation}
(a_1-a_2)(\tau\tau_{12}-\tau_1\tau_2)+D_x\tau_1\cdot\tau_2=0.
\end{equation}
\end{remark}

\subsection{Bilinear Q3}\label{sec-5-3}

The four equations in the  bilinear Q3 system \eqref{be-Q3-1} are all reflected dAKP equations.

\section{Conclusions}\label{sec-6}

In this paper, we have successfully established the connection between the bilinear ABS equations
(except Q2 and Q4) and the symmetric dAKP system.
The later is a system generated by the principle equation (3D) dAKP equation and its reflected forms.
Note that high order dAKP equations in this system are considered as the consequence of the multi-dimensional
consistency of the 3D dAKP equation.
We reviewed the bilinear ABS equations to express them in terms of a single $\tau$-function.
It is notable that for H2 equation we gave a new formulation (see \eqref{3.17}) of its bilinearization.
The present bilinear H2 system \eqref{be-f-H2} contains two equations.
Note that this bilinear H2 system can also be viewed as
the deformed bilinear B\"acklund transformation of the KdV equation \cite{CMZ-2021,CBC-CTP-2004}.
For Q1 equation with linear background, we rewrote
two bilinear equations
so that they can be presented in terms of Hirota's $D$-operator.

The symmetric dAKP system allows more high order bilinear equations which are available in implementing reductions.
As for the connection between the bilinear ABS equations and the  symmetric dAKP system,
bilinear H1 system consists of differential-difference forms of the dAKP equation and its reflected form.
Bilinear H2 system contains two bilinear equations that are differential-difference forms of two reduced dAKP equations.
H3 equation has two bilinear forms, which can be obtained from dimensional reductions
of a reflected 4D dAKP equation and two 3D dAKP equations in 4D space.
Q1 equation with power background is one of cases of H3 equation.
However, Q1 equation with linear background is more   complicated.
Two equations in its bilinear form are related to the reduced 5D dAKP and its reflected forms
and further skew-continuum limits.
Finally, bilinear Q3 equations contains four reflected 3D dAKP equations in 4D space.

As by-products, we obtained a 4D 8-point equation $M_1^\prime=0$, i.e., \eqref{A.20}
and a 3D 8-point equation $\dot{A}=0$, i.e., \eqref{alter-dAKP},
while $\dot{A}=0$ is a dimensional reduction of $M_1^\prime=0$ (see Appendix \ref{App-A}).
They have solutions composed by the dKdV type PWF.
Although at present we can not obtain $M_1^\prime=0$ as a reduction from the symmetric dAKP system,
we proved it allows Casoratian solutions (and so does the Hirota's form \eqref{tau-solution-Symmetric})
in Appendix \ref{App-A}.

As for further research, we would like to understand connections between the discrete
bilinear Boussinesq equations (see \cite{HZ-2021}) and dAKP equations.
In addition, the ABS equations, discrete Boussinesq equations and discrete KP equations
admit elliptic soliton solutions composed by Weierstrass functions \cite{NA-IMRN-2010,NijSZ-CMP-2023}.
It would be interesting to extend the connections we have found to the level of elliptic solitons.

\vskip 20pt
\subsection*{Acknowledgements}

DJZ is supported by the NSF of China (No.12271334).
JW is currently supported by the China Scholarship Council  visiting-PhD-student program.
KM is supported by JSPS KAKENHI Grant (No.22K03441).
DJZ are grateful to Prof. Hietarinta and Prof. Nijhoff for their discussions.

\appendix
\section{Equation \eqref{alter-dAKP}}\label{App-A}

\subsection{Connection with 4D equations}\label{App-A-1}

Equations in the 4D system \eqref{4D-AKP-equations-1} are featured with total index $(1,1,0,1)$.
Consider a combination
\begin{equation}
 (a_1+a_2)\mathcal{A}_1^{(1)}+(a_1-a_5)\mathcal{A}_2^{(1)}+(a_2-a_5)\mathcal{A}_3
^{(1)}=0
\end{equation}
which is expressed as
\begin{align}\label{4D-combination}
&(a_1\!+\!a_2)(a_1\!-\!a_5)(a_2\!-\!a_5)(a_3\!-\!a_5)M_1\!+\!(a_1\!+\!a_2)(a_1\!-\!a_5)(a_2\!-\!a_5)(a_3\!
+\!a_5)M_2\!+\!M_3\!=\!0,
\end{align}
where
\begin{align*}
M_1\!:=\,& (a_1\!+\!a_3)(a_1\!+\!a_5)(a_2\!-\!a_3)   \tau_{1}^3   \tau_{235}
-\!(a_1\!-\!a_2)(a_1\!-\!a_3)(a_2\!-\!a_3)   \tau^3   \tau_{1235}\nonumber\\
&-\!(a_1\!-\!a_3)(a_2\!+\!a_3)(a_2\!+\!a_5)   \tau_{2}^3   \tau_{135},\\
M_2\!:=\,& (a_1\!-\!a_3)(a_1\!+\!a_5)(a_2\!+\!a_3)   \tau_{13}  \tau_{25}^3\!
-\!(a_1\!-\!a_2)(a_1\!+\!a_3)(a_2\!+\!a_3)   \tau_{3}  \tau_{125}^3\nonumber\\
&-\!(a_1\!+\!a_3)(a_2\!-\!a_3)(a_2\!+\!a_5)   \tau_{23}   \tau_{15}^3,\\
M_3\!:=\,&2a_3(a_1-a_5)(a_1^2-a_3^2)(a_3^2-a_5^2)\{(a_1+a_2)   \tau_{2}   \tau_{15}-(a_1+a_5)
\tau_{5}   \tau_{12}+(a_2-a_5)\tau   \tau_{125}\}\nonumber\\
&-2a_3(a_2-a_5)(a_2^2-a_3^2)(a_3^2-a_5^2)\{(a_1+a_2)   \tau_{1}   \tau_{25}
+(a_1-a_5)\tau   \tau_{125}-(a_2+a_5)   \tau_{5}   \tau_{12}\}\nonumber\\
&+2a_3(a_1+a_2)(a_1^2-a_3^2)(a_2^2-a_3^2)\{(a_1\!-\!a_2)   \tau_{5}   \tau_{12}\!
-\!(a_1\!-\!a_5)   \tau_{2}   \tau_{15}\!+\!(a_2\!-\!a_5)   \tau_{1}   \tau_{25}\}.
\end{align*}
One can use equations \eqref{3D-dAKP} and \eqref{symmetric-3D-dAKP} to simplify $M_3$ to the form
\[4a_3(a_1-a_5)(a_2-a_5)(a_1^2-a_2^2)(a_3^2-a_5^2)\tau   \tau_{125},\]
such that \eqref{4D-combination} can be rewritten as
\begin{align}\label{4D-combination-1}
  &(a_3\!-\!a_5)M_1'\!+(a_3\!+\!a_5)M_2'=0,
\end{align}
where
\begin{align}
  M_1'\!:=\,&(a_1\!+\!a_3)(a_1\!+\!a_5)(a_2\!-\!a_3)   \tau_{1}^3   \tau_{235}\!
  -\!(a_1\!-\!a_2)(a_1\!-\!a_3)(a_2\!-\!a_3)   \tau^3   \tau_{1235}\!\nonumber\\
  &-\!(a_1\!-\!a_3)(a_2\!+\!a_3)(a_2\!+\!a_5)   \tau_{2}^3   \tau_{135}
  +2a_3(a_1-a_2)(a_3+a_5)\tau   \tau_{125},\label{M1}\\
  M_2'\!:=\,&(a_1\!-\!a_3)(a_1\!+\!a_5)(a_2\!+\!a_3)   \tau_{13}   \tau_{25}^3\!
  -\!(a_1\!-\!a_2)(a_1\!+\!a_3)(a_2\!+\!a_3)   \tau_{3}   \tau_{125}^3\!\nonumber\\
  &-\!(a_1\!+\!a_3)(a_2\!-\!a_3)(a_2\!+\!a_5)   \tau_{23}   \tau_{15}^3
  +2a_3(a_1-a_2)(a_3-a_5)\tau   \tau_{125}.\label{M2}
\end{align}
Noted that $M_2'$ is a reflected form of $M_1'$ by taking
\begin{equation}
\tau(n_1,n_2,(N_3,a_3),n_5)\!=\!\tau(n_1,n_2,(-N_3,-a_3),n_5).
\end{equation}
By employing constraint
\begin{equation}\label{A.8}
a_3=a_5,~~\tau(n_1,n_2,N_3,n_5)\rightarrow\tau(n_1,n_2,n_3^\prime=N_3+n_5)
\end{equation}
 on $M_1'$ \eqref{M1} and $M_2'$ \eqref{M2},
they will  reduce  to
\begin{equation}
  \begin{aligned}\label{dot-A-1}
   & (a_2 -a_3)(a_1 + a_3)^2\tau_{1}^3\tau_{233}- (a_1 - a_2) (a_1 - a_3) (a_2 - a_3)\tau^3\tau_{1233} \\
    & -(a_1 - a_3)(a_2 +a_3)^2  \tau_{2}^3 \tau_{133} +4a_3^2  (a_1 - a_2) \tau\tau_{123}
  \end{aligned}
\end{equation}
  and
  \begin{equation}\label{red-M2}
     (a_1+a_3)(a_2+a_3)\{(a_1-a_3) \tau_{13}    \tau_{2}-(a_1-a_2)   \tau_{12} \tau_{3}
     -(a_2-a_3) \tau_{23}   \tau_{1}\}
  \end{equation}
respectively, while \eqref{dot-A-1} is nothing but $\dot{A}$ \eqref{alter-dAKP} and
\eqref{red-M2} contains the dAKP \eqref{3D-dAKP}.

In the next part we will prove both $M_1'$   and $M_2'$ are zero, and so is $\dot{A}$.

\subsection{Proof of solution}\label{App-A-2}

For convenience, we introduce some notations related to the generic column vector
$\psi=\psi(n_1,n_2,\cdots,n_r;l_1)=(\psi_1,\psi_2,\cdots,\psi_N)^T$.
We define
\begin{align*}
 &  m_{\substack{n_i+1}}:=E_{l_1}^mE_{n_i}\psi,~~
 m_{\substack{n_i+1\\n_j-1}}:=E_{l_1}^mE_{n_j}^{-1}E_{n_i}\psi,\\
  &\widetilde{m}_{n_i+j}^{n_{\xi}}:=\Theta_{n_\xi}^{-1} E_{l_1}^mE_{n_i}^j\psi,~~~
  \widetilde{m}_{n_i+j}^{n_{\xi},n_{\theta}}\!:
  =\Theta_{n_\xi}^{-1}\Theta_{n_\theta}^{-1} E_{l_1}^mE_{n_i}^j\psi,
  \end{align*}
where
\begin{equation}
\Theta_{n_\xi}=\mathrm{Diag}(U_1(n_{\xi}), U_2(n_{\xi}), \cdots, U_N(n_{\xi})),~~
U_i(n_{\xi})=k_i^2-a_{\xi}^{2}.
\end{equation}
In the proof, we also replace $N_3^\prime$ with $n_3$ without making confusion.
\begin{theorem}\label{th-7}
The equations
\begin{equation}
M_1'=0,~~     M_2'=0,~~  \dot{A}=0
\end{equation}
share the $\tau$-function
\begin{equation}\label{A.14}
  \tau=\bigl\lvert \widehat{N-1}\bigr\rvert_{[n_\sigma]}= \bigl\lvert \widehat{N-1}\bigr\rvert_{[l_1]},
  ~~ (\sigma=1,2,3,5)
  \end{equation}
composed by $\psi=(\psi_1,\psi_2,\cdots,\psi_N)^T$  with entries
  \begin{align}\label{DR-AKP}
    \psi_j(n_1,n_2,n_3,n_5;l_1 )=&\rho_j^+ k_i^{l_1}(a_1+k_j)^{n_1}(a_2+k_j)^{n_2}
    (a_3+k_j)^{n_3}(a_5+k_j)^{n_5}\nonumber\\
    &+\rho_j^- (-k_j)^{l_1}(a_1-k_j)^{n_1}(a_2-k_j)^{n_2}(a_3-k_j)^{n_3}(a_5-k_j)^{n_5}.
  \end{align}
\end{theorem}

\begin{proof}
The entry $\psi_j$ satisfies
\begin{align*}
& a_3\psi_j(n_3+1;l_1)-\psi_j(n_3+1;l_1+1)=(a_3^2-k_j^2)\psi_j(n_3;l_1),\label{r-1}\\
& a_5\psi_j(n_5+1;l_1)-\psi_j(n_5+1;l_1+1)=(a_5^2-k_j^2)\psi_j(n_5;l_1), \\
& a_i \psi_j(n_i-1;l_1)+\psi_j(n_i-1;l_1+1)=\psi_j(n_i;l_1), \\
&(a_3+a_i)\psi_j(n_3+1,n_i-1;l_1)=\psi_j(n_3+1,n_i;l_1)+(a_3^2-k_j^2)\psi_j(n_3,n_i-1;l_1), \\
&(a_5+a_i)\psi_j(n_5+1,n_i-1;l_1)=\psi_j(n_5+1,n_i;l_1)+(a_5^2-k_j^2)\psi_j(n_5,n_i-1;l_1), \\
& (a_i-a_r)\psi_j(n_i-1,n_r-1;l_1)=\psi_j(n_i,n_r-1;l_1)-\psi_j(n_i-1,n_r;l_1), \\
&(a_5\!-\!a_3)\psi_j(n_3\!+\!1,n_5\!+\!1;l_1)\!=\!(a_5^2\!-\!k_j^{2})\psi_j(n_3\!
+\!1,n_5;l_1)\!-\!(a_3^2\!-\!k_j^2)\psi_j(n_3,n_5\!+\!1;l_1),
\end{align*}
which give rise to (note that because of \eqref{A.14} here the Casoratians are displayed
with respect to column shifting in $l_1$-direction)
\begin{align*}
& \tau^{i}=\left\lvert \psi(n_i-1;0)~~\psi(n_i-1;1)~~\cdots~~\psi(n_i-1;N-1)\right\rvert\nonumber\\
  &~~~=\left\lvert \psi(n_i-1;0)~~\psi(n_i;0)~~\cdots~~\psi(n_i;N-2)\right\rvert\nonumber\\
&~~~=\bigl\lvert 0_{n_i-1}~~0~~\cdots~~N-2\bigr\rvert
=\bigl\lvert 0_{n_i-1}~~\widehat{N-2}\bigr\rvert,\\
  &\tau^{ij}=\Bigl\lvert 0_{n_i-1\atop n_j-1}~~0_{n_j-1}~~\widehat{N-3}\Bigr\rvert
  =\Bigl\lvert 0_{n_i-1\atop n_j-1}~~0_{n_i-1}~~\widehat{N-3}\Bigr\rvert,\\
  &\tau^{ijr}=\Biggl\lvert 0_{\substack{n_i-1\\n_j-1\\n_r-1}}~~
  0_{n_i-1\atop n_j-1}~~0_{n_j-1}~~\widehat{N-4}\Biggr\rvert,\\
  &\tau_{3}=\Gamma(a_3) \bigl\lvert \widetilde{0}_{n_3+1}^{n_3}~~\widehat{N-2}\bigr\rvert,~~~~
  \tau_{33}=(\Gamma(a_3))^2
  \bigl\lvert \widetilde{0}_{n_3+2}^{n_3,n_3}~~\widetilde{0}_{n_3+1}^{n_3}~~\widehat{N-3}\bigr\rvert,\\
&  \tau_{35}=\Gamma(a_3)\Gamma(a_5)
\Bigl\lvert \widetilde{0}_{\substack{n_3+1\\n_5+1}}^{n_3,n_5}~~
\widetilde{0}_{n_5+1}^{n_5}~~\widehat{N-3}\Bigr\rvert,\\
&\tau_{35}^j=\Gamma(a_3)\Gamma(a_5)
  \Biggl\lvert \widetilde{0}_{\substack{n_3+1\\n_5+1\\n_j-1}}^{n_3,n_5}~~
  \widetilde{0}_{\substack{n_5+1\\n_j-1}}^{n_5}
  ~~0_{n_j-1}~~\widehat{N-4}\Biggr\rvert,\\
 &(a_i-a_j) \tau^{ij}=\left\lvert 0_{n_j-1}~~0_{n_i-1}~~\widehat{N-3}\right\rvert,\\
 & (a_3+a_j)(a_5+a_j)\tau_{35}^j=\Gamma(a_3)\Gamma(a_5)
 \left\lvert \widetilde{0}_{\substack{n_3+1\\n_5+1}}^{n_3,n_5}
 ~~\widetilde{0}_{n_5+1}^{n_5}~~0_{n_j-1}
 ~~\widehat{N-4}\right\rvert,\\
&~~~~~~~~~~~~~~~~~~~~~~~~~~~~~~\cdots\cdots\nonumber
\end{align*}
where $\Gamma(a_j)=|\Theta_{n_j}|$.
Accordingly, $E_{n_2}^{-1}E_{n_1}^{-1}M_1^{'}$ turns to
\begin{equation*}
  \begin{split}
  &\Gamma(a_3)\Gamma(a_5)\left(\bigl\lvert 0_{n_3-1}~~0_{n_2-1}~~\widehat{N-3}\bigr\rvert
  \left\lvert 0_{n_1-1}~~\widetilde{0}_{\substack{n_3+1\\n_5+1}}^{n_3,n_5}
  ~~\widetilde{0}_{n_5+1}^{n_5}~~\widehat{N-4}\right\rvert\right.\\
 &-\left\lvert \widetilde{0}_{\substack{n_3+1\\n_5+1}}^{n_3,n_5}~~\widetilde{0}_{n_5+1}^{n_5}
 ~~\widehat{N-3}\right\rvert
 \bigl\lvert 0_{n_3-1}~~0_{n_2-1}~~0_{n_1-1}~~\widehat{N-4}\bigr\rvert\\
  &-\bigl\lvert 0_{n_3-1}~~0_{n_1-1}~~\widehat{N-3}\bigr\rvert
  \left\lvert 0_{n_2-1}~~\widetilde{0}_{\substack{n_3+1\\n_5+1}}^{n_3,n_5}~~\widetilde{0}_{n_5+1}^{n_5}
  ~~ \widehat{N-4}\right\rvert\\
  &\left.+\bigl\lvert 0_{n_2-1}~~0_{n_1-1}~~\widehat{N-3}\bigr\rvert
  \left\lvert 0_{n_3-1}~~\widetilde{0}_{\substack{n_3+1\\n_5+1}}^{n_3,n_5}~~\widetilde{0}_{n_5+1}^{n_5}
  ~~\widehat{N-4}\right\rvert\right),
  \end{split}
\end{equation*}
which is the expansion of following determinant
\begin{equation}
\begin{split}
&\Gamma(a_3)\Gamma(a_5)\,\left | \begin{matrix}
  \widehat{N-4} & N-3 & 0_{n_3-1}& 0_{n_2-1} & 0_{n_1-1} & 0 & 0 & 0 \\
  0 & N-3 &0_{n_3-1} & 0_{n_2-1} & 0_{n_1-1} & \widetilde{0}_{\substack{n_3+1\\n_5+1}}^{n_3,n_5} &
  \widetilde{0}_{n_5+1}^{n_5} & \widehat{N-4}
  \end{matrix} \right |,
\end{split}
\end{equation}
which is zero.
Hence we have
\begin{equation}
\begin{aligned}
  M_1'\!:=~&\!(a_1\!+\!a_3)(a_1\!+\!a_5)(a_2\!-\!a_3)   \tau_{1}^3   \tau_{235}\!
  -\!(a_1\!-\!a_2)(a_1\!-\!a_3)(a_2\!-\!a_3)   \tau^3   \tau_{1235}\!\\
  &-\!(a_1\!-\!a_3)(a_2\!+\!a_3)(a_2\!+\!a_5)   \tau_{2}^3   \tau_{135}+2a_3(a_1-a_2)(a_3+a_5)\tau   \tau_{125}=0.
\end{aligned}\label{A.20}
\end{equation}
Similarly, for $E_{n_2}^{-1}E_{n_1}^{-1}M_2^{'}$, we have
\begin{equation*}
  \begin{split}
  &\Gamma(a_3)\Gamma(a_5)\left(
  \bigl\lvert \widetilde{0}_{n_3+1}^{n_3}~~0_{n_2-1}~~\widehat{N-3}\bigr\rvert
  \left\lvert \widetilde{0}_{\substack{n_5+1\\n_3-1}}^{n_5}~~0_{n_3-1}~~0_{n_1-1}~~
  \widehat{N-4}\right\rvert\right.\\
 &-\left\lvert \widetilde{0}_{\substack{n_5+1\\n_3-1}}^{n_5}~~0_{n_3-1}~~\widehat{N-3}\right\rvert
 \bigl\lvert \widetilde{0}_{n_3+1}^{n_3}~~0_{n_2-1}~~0_{n_1-1}~~\widehat{N-4}\bigr\rvert\\
  &-\bigl\lvert \widetilde{0}_{n_3+1}^{n_3}~~0_{n_1-1}
  ~~\widehat{N-3}\bigr\rvert\left\lvert \widetilde{0}_{\substack{n_5+1\\n_3-1}}^{n_5}
  ~~0_{n_3-1}~~0_{n_2-1}~~\widehat{N-4}\right\rvert\\
  &\left.+\bigl\lvert 0_{n_2-1}~~0_{n_1-1}~~\widehat{N-3}\bigr\rvert
  \left\lvert \widetilde{0}_{n_3+1}^{n_3}~~\widetilde{0}_{\substack{n_5+1\\n_3-1}}^{n_5}~~0_{n_3-1}
  ~~\widehat{N-4}\right\rvert\right),
  \end{split}
\end{equation*}
which corresponds to following determinant
\begin{equation}
  \begin{split}
  &\Gamma(a_3)\Gamma(a_5)\,\left | \begin{matrix}
    \widehat{N-4} & N-3 & 0 & 0_{n_2-1} & 0_{n_1-1} & 0 & \widetilde{0}_{n_3+1}^{n_3} & 0 \\
    0 & N-3 &0_{n_3-1} & 0_{n_2-1} & 0_{n_1-1} & \widetilde{0}_{\substack{n_5+1\\n_3-1}}^{n_5}
    & \widetilde{0}_{n_3+1}^{n_3} & \widehat{N-4}
    \end{matrix} \right |,
  \end{split}
  \end{equation}
which is zero as well.
So we have
  \begin{equation}
  \begin{aligned}\label{M2-equation}
    M_2'\!:=~&\!(a_1\!-\!a_3)(a_1\!+\!a_5)(a_2\!+\!a_3)
    \tau_{13}   \tau_{25}^3\!-\!(a_1\!-\!a_2)(a_1\!+\!a_3)(a_2\!+\!a_3)   \tau_{3}   \tau_{125}^3\!\\
    &-\!(a_1\!+\!a_3)(a_2\!-\!a_3)(a_2\!+\!a_5)   \tau_{23}   \tau_{15}^3+2a_3(a_1-a_2)(a_3-a_5)\tau \tau_{125}=0.
  \end{aligned}
\end{equation}
Finally,   $E_{n_2}^{-1}E_{n_1}^{-1}\dot{A}$ is converted to
\begin{equation*}
  \begin{split}
& (\Gamma(a_3))^2 \left(\bigl\lvert 0_{n_3-1}~~0_{n_2-1}~~\widehat{N-3}\bigr\rvert
\left\lvert 0_{n_1-1}~~ \widetilde{0}_{\substack{n_3+2}}^{n_3,n_3}~~\widetilde{0}_{\substack{n_3+1}}^{n_3}
~~
\widehat{N-4}\right\rvert\right.\\
&- \bigl\lvert \widetilde{0}_{n_3+2}^{n_3,n_3}~~\widetilde{0}_{n_3+1}^{n_3}~~\widehat{N-3}\bigr\rvert
\left\lvert 0_{\substack{n_3-1}}~~0_{n_2-1}~~0_{n_1-1}~~\widehat{N-4}\right\rvert\\
&-\bigl\lvert 0_{n_3-1}~~0_{n_1-1}~~\widehat{N-3}\bigr\rvert
\left\lvert 0_{n_2-1}~~\widetilde{0}_{\substack{n_3+2}}^{n_3,n_3}~~\widetilde{0}_{\substack{n_3+1}}^{n_3}
~~\widehat{N-4}\right\rvert\\
&\left.+\bigl\lvert 0_{n_2-1}~~0_{n_1-1}~~\widehat{N-3}\bigr\rvert
\left\lvert 0_{n_3-1}~~\widetilde{0}_{n_3+2}^{n_3,n_3}~~\widetilde{0}_{n_3+1}^{n_3}
~~\widehat{N-4}\right\rvert\right),
  \end{split}
\end{equation*}
which is equivalent to the the zero valued determinant
\begin{equation}
  \begin{split}
  &(\Gamma(a_3))^2 \,\left | \begin{matrix}
    \widehat{N-4} & N-3 & 0_{n_3-1}& 0_{n_2-1} & 0_{n_1-1} & 0 & 0 & 0 \\
    0 & N-3 &0_{n_3-1} & 0_{n_2-1} & 0_{n_1-1} & \widetilde{0}_{\substack{n_3+2}}^{n_3,n_3}
    & \widetilde{0}_{n_3+1}^{n_3} & \widehat{N-4}
    \end{matrix} \right |.
  \end{split}
  \end{equation}
Thus we arrive at
\begin{equation}\label{A.21}
\begin{aligned}
  \dot{A}:=& (a_2 -a_3)(a_1 + a_3)^2\tau_{1}^3\tau_{233}- (a_1 - a_2) (a_1 - a_3) (a_2 - a_3)\tau^3\tau_{1233} \\
  & -(a_1 - a_3)(a_2 +a_3)^2  \tau_{2}^3 \tau_{133} +4a_3^2  (a_1 - a_2) \tau\tau_{123}=0.
\end{aligned}
\end{equation}
Note that
it is also a consequence
of \eqref{A.20} together with the reduction \eqref{A.8}
that \eqref{A.21} holds.

\end{proof}

\end{document}